\documentclass[12pt,reqno]{amsart}
\usepackage[utf8]{inputenc}
\usepackage{amsmath,amssymb,amsthm,hyperref,array,xcolor,mathtools,multicol,subcaption,verbatim,microtype}
\hypersetup{
pdftitle={Mechanism Design by a Politician}, 
pdfauthor={Giovanni Valvassori Bolgè}, 
pdfkeywords={Legislative Coalitions, Public Goods, Mechanism Design, Reservation Utility}, 
pdfcreator={LaTeX with hyperref package}, 
pdfproducer={LaTeX},
colorlinks=true,
linkcolor=[rgb]{0,0,0.6}, 
citecolor=[rgb]{0,0,0.6},  
urlcolor=[rgb]{0,0,0.6},   
}

\usepackage[normalem]{ulem}
\usepackage[pdftex]{graphicx}
\usepackage{fullpage}
\usepackage{cleveref}
\usepackage{dsfont}
\usepackage{tikz}
\usepackage{tikz-3dplot}
\usetikzlibrary{patterns,patterns.meta,arrows.meta, calc}

\usepackage{enumitem}
\setlist[itemize]{nosep}
\setlist[enumerate]{nosep}

\usepackage{natbib}

\newtheorem*{theorem*}{Theorem}
\newtheorem{theorem}{Theorem}
\newtheorem{lemma}{Lemma} 
 
\newtheorem{definition}{Definition} 
\newtheorem{corollary}{Corollary} 
\newtheorem{assumption}{Assumption} 
\newtheorem{proposition}{Proposition} 
 
\newtheorem{claim}{Claim}

\theoremstyle{remark}

\DeclareMathOperator*{\argmax}{arg\,max}
\DeclareMathOperator*{\argmin}{arg\,min}

\newcommand{\vbar}{\bar{v}}
\newcommand{\ubar}[1]{\underline{#1}}
\newcommand{\dd}{\,\mathrm{d}}

\newcommand{\ct}{\tilde{\theta}}
\usepackage{accents}

\crefname{equation}{}{}

\begin{document}
	
\title{Mechanism Design by a Politician*}
\thanks{*I thank Andreas Asseyer, Berno B\"{u}chel, Laura Doval, Piotr Dworczak, Tan Gan, Susanne Goldl\"{u}cke, Martin Hagen, Andrew Mackenzie, Roger B.  Myerson, Mariann Ollar, Fikri Pitsuwan, Marek Pycia, Florian Scheuer, Yves Sprumont, Julian Teichgr\"{a}ber, as well participants at the University of Fribourg, University of Bielefeld, University of St. Gallen, the 2024 Coalition Theory Network Workshop in Budapest, the 2024 SGVS Conference, the 2024 SIOE Conference in Chicago, the 2024 Conference on Mechanism and Institution Design in Budapest, the 2025 Conference on Economic Design in Essex, the 2025 SAET Conference in Ischia, the 2025 Stony Brook Conference on Game Theory.
}
	 
\author{Giovanni Valvassori Bolgè\textsuperscript{\textdagger}}\thanks{\textsuperscript{\textdagger}University of Fribourg, Department of Economics. \textit{Email: \href{mailto:giovanni.valvassoribolge@unifr.ch}{giovanni.valvassoribolge@unifr.ch}}}
	
\date{\today}

\begin{abstract} 

	A set of agents has to make a decision about the provision of a public good and its financing. Agents have heterogeneous values for the public good and each agent's value is private information. An agenda-setter has the right to make a proposal about a public-good level and a vector of contributions. For the proposal to be approved, only the favourable votes of a majority of agents are needed. If the proposal is not approved, a type-dependent outside option is implemented. I characterize the optimal public-good provision and the coalition-formation for any outside option in dominant strategies. Optimal public-good provision might be a non-monotonic function of the outside option public-good level. Moreover, the optimal coalition might be a non-convex set of types. Importantly, majority voting creates endogenously countervailing incentives among agents.
	
	
\end{abstract}

\maketitle

\newpage

\section{Introduction}

This paper revisits a classic mechanism design problem, the provision of a public good, by considering a political economy environment. On the one hand, I aim at contributing to the theory of mechanism design by providing a novel political economy perspective. On the other hand, I apply insights from the theory of screening to study foundational topics in political economy, namely public-good provision and coalition formation.

Political economy developed at the intersection of Economics and Political Science as an attempt to incorporate political dynamics into the analysis of economic policy. In the words of \cite{alesina94}:
\begin{quote}
    ``Political-economy models begin with the assertion that economic policy choices are not made by social planners, who live only in academic papers. Rather, economic policy is the result of political struggle within an institutional structure.''
\end{quote}
Furthermore, drawing on the developments in the theory of screening\footnote{For a recent and thorough overview, see \citet{börgers2015}.}, an alternative view of political economy has been put forth by the work of Jean-Jacques Laffont and coauthors.\footnote{For an extensive treatment, see \cite{laffont2001}.} The Constitution-politician agency problem is reminiscent of the principal-agent problem in the theory of contracts. As it was made clear in \cite{laffont2000b},
\begin{quote}
    ``[p]olitical economy is the recognition of the incompleteness of the Constitution [...]. Because the Constitution is an incomplete contract, politicians have a lot of discretion, that they use to further interest groups."
\end{quote}
In particular, the incompleteness of the Constitution entails the absence of adequate institutional tools to cope with asymmetric information. In fact, according to \cite{laffont2001}:
\begin{quote}
    ``... a major inefficiency of political conflicts follows from the inefficiency of redistributive instruments due to asymmetric information. [...] Politics is then a game of redistribution of information rents within the realm of discretion left by the Constitution."
\end{quote}
In this paper, I consider a public-good provision model with privately-informed, heterogeneous agents where the mechanism designer is one of the agents. In particular, this agent, whom I refer to as the `politician', is the agenda-setter in a collective choice problem with three main features.

First, the collective choice entails the provision of a nonrival, nonexcludable public good and its financing. Second,
the contributions borne by different agents are allowed to be different. Third, in order for the proposal to be implemented, the agenda-setter only needs the support of a subset of the agents (e.g., simple or supermajority).

If the proposal is not approved, a type-dependent outside option is implemented.
Drawing from insights in the theory of screening, I characterize optimal public-good provision and coalition formation for any outside option.
I focus on dominant-strategy implementation.

The main idea is a fundamental connection between monopolistic screening and agenda-setting. When there is a single agent, the two problems are identical. Like a monopolist, the agenda-setter has to elicit from the agent his/her value for the public-good and offer him/her a contract involving a quantity (in this case, a public-good level) in exchange for a transfer. When there are multiple agents, the two problems are still mathematically equivalent \textit{ex ante}. However, the \textit{ex post} constraints are fundamentally different.

In fact, whereas in monopolistic screening the quantity of the good traded is consumed privately by each agent, the public good is nonrival and nonexcludable; hence, once a public-good level has been chosen, all agents will consume it, regardless of the individual contributions to its financing.

Building on key contributions in monopolistic screening with type-dependent outside options (\cite{lewis1989}, \cite{maggi1995}, \cite{jullien2000}), I show that the shape of the reservation utility profile, namely the utility profile yielded by the outside option, crucially determines the optimal contract, and hence public-good provision and coalition formation.

From a methodological perspective, I develop a novel technique to solve a monopolistic screening problem with type-dependent outside option without relying on optimal control methods.

Intuitively, for a given outside option, once a public-good level has been provided, some agents will have the \textit{ex post} incentive to understate their preference for the public good, whereas others will have the \textit{ex post} incentive to overstate it. If the agenda-setter can provide a public-good level which is consistent with all these constraints, then the public-good provision will differ from the outside option. Otherwise, the outside option will be implemented.

When the agenda-setter only needs a subset of favourable votes for the proposal to be implemented, some agents can be excluded from the winning coalition and therefore their participation constraints can be violated.\footnote{If the incentive compatibility constraints are required to hold only for the members of the coalition, then it is optimal to randomly choose the agents and propose the mechanism to $q-1$ of them, while taxing the others to some upper bound $\Bar{\tau}$.}

If the incentive compatibility constraints have to hold for all agents, including those left out of the winning coalition, then also this set of excluded agents will affect the optimal public-good provision, since those agents still have to be given the incentive to truthfully reveal their value for the public good.

Interestingly, the winning coalition varies significantly with the outside option and the shape of the reservation utility profile. When the outside option public-good level is sufficiently low (high), it is optimal to choose the agents with the `highest' (`lowest') value for the public good, \textit{regardless} of the shape of the reservation utility profile. Moreover, when sufficiently many agents can be excluded from the coalition, public-good provision is non-monotonic in the outside option public-good level.

On the other hand, for intermediate values of the outside option public-good level, the winning coalition crucially depends on the concavity or convexity of the reservation utility profile. When it is concave, the winning coalition can be a \textit{non-convex} set of agents, comprising agents with high and low values for the public good, while excluding other agents with intermediate values. When the reservation utility profile is convex, the winning coalition is a convex set of agents, excluding either (or both) the agents with the highest, and lowest, value for the public good.

Optimal public-good provision is indeed the result of a `game of redistribution of information rents', as pointed out by \cite{laffont2001}.

The rest of the paper is structured as follows. \Cref{sec:lite} reviews the literature. \Cref{sec:model} introduces the model.  \Cref{sec:sb} lays out the analysis. \Cref{sec:linear} foreshadows the main results, which are then presented in \Cref{sec:main}. \Cref{sec:implications} characterizes the implications of the main results. \Cref{sec:relaxss} discusses the relaxations of the assumptions and \Cref{sec:conclusions} concludes. All proofs are in \Cref{sec:appendix}.

\section{Applications}
The model of this paper applies to collective choice problems with three features: \textit{(i)} a proposer has agenda-setting power; \textit{(ii)} a voting threshold $q$; \textit{(iii)} once the required majority is obtained, the outcome is binding for all agents. Three applications are particularly relevant.
\subsection*{Participatory Budgeting} As a democratic procedure devised to identify and implement public projects, Participatory Budgeting originated in Brazil and then quickly spread over several municipalities all over the world.\footnote{For a comprehensive treatment at the intersection of Computer Science and Economics, see \cite{aziz2021participatory}.}
The core idea is for a set of individuals, or entities, to propose public projects, as well as their financing, so as to improve accountability and efficiency of public resources. After possibly many rounds of discussion, a proposal will be selected. If accepted through a voting stage, the project is implemented and the municipality bears the cost of its financing. Importantly, even the individuals who voted against the proposal will contribute to the financing of the project through taxation.

This paper provides a theoretical foundation for a fundamental question: how will informational costs affect the scale of the public project and what coalition should the agenda-setter target in order for the project to be implemented?

\subsection*{Organizational Economics} Consider an organization who has to take a decision over a policy which is irreversible (at least in the short-run), such as the election of a new CEO or the delocalization of management and production to a different country. 
The set of individuals who favour the move will advance a proposal to the board of directors, and a vote will take place. If the relevant threshold of favourable votes will be reached, then the firm will implement the policy. Crucially, all individuals will bear the consequences of the policy, whether or not they voted in favour.

\subsection*{Blockchain Economics} In on-chain
governance systems, protocol upgrades are adopted by
token-weighted votes subject to explicit quorum and supermajority
thresholds. An adopted
upgrade binds \emph{every} holder, since the chain's state is
common; a dissenter's only exit is a fork, whose value is the
legacy protocol degraded by forgone network effects (\citet{biais2019blockchain}).
The proposing foundation or core-developer team is the
agenda-setter; and the transfer instrument, elsewhere an
abstraction, is given by programmable member-specific payments,
which can be executed by the same smart contracts that implement the upgrade
\citep{cong2019blockchain,catalini2020simple}.\footnote{The security economics of the thresholds themselves ---
what majorities can and cannot be trusted with --- is the theme
of \citet{budish2025trust}; \citet{halaburda2022microeconomics}
survey the field.}

\section{Related Literature} \label{sec:lite}
This paper contributes to the literature on public-good provision
at the intersection of political economy and mechanism design.
 
\subsection*{Legislative bargaining.} In political economy, an
established literature develops a positive theory of public-good
provision through legislative bargaining
\citep{baron1989,battaglini2007,battaglini2008,bowen2014}. These
contributions study bargaining among heterogeneous players under
complete information; here, by contrast, agents are privately
informed about their valuations, and the analysis characterizes
how the resulting information rents shape both public-good provision and coalition formation.
 
\subsection*{Mechanisms for public goods.} Within mechanism design,
this paper relates to two strands of the literature.

The first studies dominant-strategy first-best provision through variants of the VCG mechanism.\footnote{See \citet{green1979a} for a treatise, as well as \citep{kuzmics2017} for
recent work on budget-balanced dominant-strategy provision. For bilateral trade, see
\citet{hagerty_robust_1987} and \citet{hagen_multidimensional_2021}.}
The second strand studies \emph{optimal} --- objective-maximizing
--- provision under \emph{voluntary} participation, where a
fundamental negative result holds: with independent private
values, individual rationality, and budget balance, provision
collapses as the economy grows --- the free-rider problem is, at
bottom, a participation problem
\citep{guth1986,rob1989,mailath1990,hellwig2003}.

This paper
studies majority rule with forced participation. The
constitution, not the transfer scheme, relaxes individual
rationality --- and the design problem becomes \emph{whom} to
coerce. Relative to \citet{norman2004}, where efficiency is
restored by excluding agents from \emph{consumption}, exclusion
here operates on \emph{consent}: the good remains non-rival and
non-excludable in use, but some agents vote against the agenda-setter's proposal.
 
\subsection*{Three departures.} Against both strands, the paper
departs in three regards. First, the designer is one of the
agents: privately informed, self-interested, and constrained only
by the institution. Second, a subset of
agents can be bound by the decision even when their participation constraints are violated. Third, the size of the
project is endogenous to coalition formation, not fixed
\emph{a priori}.
 
\subsection*{Type-dependent outside options.} The paper's screening
structure builds on the literature on countervailing incentives
and type-dependent participation (\citet{baron1982a},
\citet{lewis1989}, \citet{maggi1995}, \citet{jullien2000}). \citet{figueroa_role_2009} study
type-dependent outside options in auctions, and
\citet{teichgraber2022} apply the framework to short-time work
schemes. Those models contract over a private good with a single
agent; here the designer provides a non-rival, non-excludable
public good to $n$ agents under \emph{ex post} constraints, so
the participation structure --- which types bind, and on which
sections of the type space --- is jointly determined with the
coalition.

Recently, \citet{dworczak2025property} consider the outside option itself as the design object, namely the optimal structure of property rights. In this paper, the outside option is given, and the agenda-setter has to choose the public-good provision and the winning coalition.
 
\subsection*{Political economy and mechanism design.} At the
intersection of political economy and mechanism design, \citet{bierbrauer2016} show that robustly
coalition-proof mechanisms for public goods are voting
mechanisms. This paper is complementary, as it takes the voting institution as given
and asks what a strategic agenda-setter finds optimal to do with it.
\citet{scheuer} study capital taxation when
policy credibility requires the support of an endogenous
political coalition. To the best of my knowledge, the present
paper is the first to embed coalition \emph{formation} --- the
choice of whose participation constraints to honor --- inside an
optimal mechanism design problem.
 
\subsection*{Closest contributions.} The closest papers are
\citet{dragu2017} and \citet{pitsuwan2023}.

\citet{dragu2017}
adopt a mechanism-design approach to legislative coalition
formation with privately informed politicians, showing that
equilibrium coalitions are center-left or center-right according
to the median party's position relative to the outside option.
The present paper departs in two regards.

First, the
agenda-setter is him/herself a player who controls the resource
constraint. Second, whereas
\citet{dragu2017} normalize the outside option to zero, the
solution here is characterized for \emph{any} reservation
profile, and the shape of $\vbar_i$ together with the level of
$g^{\circ}$ generates qualitatively different coalitions ---
including non-convex ones.

\citet{pitsuwan2023} study a similar
model under complete information; private information reverses or
refines their conclusions in three ways.

First, extreme
coalitions remain optimal only for sufficiently high or low
outside-option levels; at intermediate levels, the optimal
coalition can be a non-convex set of types, depending on the shape of the reservation profile. Second, in their setting, it is never optimal to propose the outside option. In this setting, it might be optimal to do so when the reservation utility profile is linear in types. Third, the reservation utility profile considered in this paper is more general than the one studied in \cite{pitsuwan2023}.

\section{Model} \label{sec:model}
There is a set $N=\{1, \dots, n\}$ of agents, with $n \geq 2$. They have to make a collective choice about the provision of a nonrival, nonexcludable public good and its financing. Agents are heterogeneous with respect to their value for the public good. For each agent $i \in N$, the value for the public good (or the `type') $\theta_i$ is private information. In particular, $\theta_i \in \Theta_i \equiv [\underline{\theta}, \overline{\theta}]$ has a cumulative distribution function $F_i$ with strictly positive density $f_i(\theta_i)>0$. Denote by $\theta$ the generic vector $(\theta_1, \theta_2,..., \theta_n)$, where $\theta \in \Theta \equiv [\underline{\theta}, \overline{\theta}]^{N}$. Moreover, denote $f(\theta)= \prod_{i \in N}f_i(\theta_i)$.\\
The agenda-setter, $a \in N$, makes a proposal consisting of a public-good level and a vector of contributions. Formally, the proposal is a $n+1$-tuple $(g,t_i)_{i \in N}$, where $g \in \mathbb{R}_+$ is the public-good level and $t_i \in \mathbb{R}$ is the contribution of agent $i$. The cost of producing the public good is linear: $c(g)=g$. The proposal has to satisfy a resource constraint
\begin{equation*}
    g \leq \sum_{i \in N} t_i \ . \\
\end{equation*}
In other words, the total contributions by the agents need to be at least as high as the cost of providing the public good. Given a policy proposal, agent $i$'s utility is given by
\begin{equation*}
    u_i(\theta_i, g,t_i)=\theta_i \phi(g)-t_i
\end{equation*}
with $\phi : \mathbb{R}_+ \longrightarrow \mathbb{R}_+$ being a twice-differentiable, non-decreasing and strictly concave function. Moreover, it satisfies $\phi(0)=0$.\\
The agenda-setter needs the support of $q\leq n$ agents to have his/her policy implemented. If the proposal is not approved, then an outside option public-good level is implemented. This yields utility $\Bar{v}_i(\theta_i,g^\circ)$, where $\Bar{v}_i$ is non-decreasing in $g^\circ$, twice-differentiable and it satisfies $\Bar{v}_i(\cdot, 0)=0$. Importantly, $\Bar{v}_i(\theta_i, g^\circ)$ is type-dependent. Moreover, the utility levels from the agenda-setter's proposal may differ from those under the outside option. This generality is meant to capture some potential irreversibilities entrenched in the institutional framework.

Participatory budgeting or the offshoring of a firm to a different country constitute examples which were discussed above.

In addition to these, a set of local governments might want to provide a nation-wide public-good provision. In this case, the provision of the public good at the national level generates utility levels which are different under the outside option. This might be due to the presence of fixed costs or the impossibility of providing the same public-good level at the same cost with the local technology.\footnote{Caveat: different technologies might imply different utility function $\Bar{v}_i$.}

Alternatively, the outside option might be linked to the marginal benefit of the policy proposal in a nonlinear way. For example, some agents might benefit from the policy proposal the most because they would be the worst-off under the outside option.

Therefore, no restriction is placed on the shape of the reservation utility profile.
\medskip\\

\section{Second-best Mechanisms} \label{sec:sb}
We now consider second-best mechanisms.\footnote{\Cref{sec:fb} discusses first-best mechanisms.} By the Revelation Principle (\cite{myerson1979a}), we can restrict our attention to direct mechanisms without loss of generality.
\begin{definition} \label{def:direct}
A direct mechanism is a set of functions $g$ and $t_i$, \quad  $\forall i \in N$
\begin{align*}
    g &: \Theta \longrightarrow \mathbb{R}_+,\\
    t_i &: \Theta \longrightarrow \mathbb{R} \ .
\end{align*}
\end{definition}
\medskip

Formally, there is the following sequence of events.

\begin{enumerate}
    \item An agenda-setter $a \in N$ has the right to make a proposal.
    \item Each agent $i \in N$ is privately informed about his/her own value for the public-good $\theta_i$.
    \item The agenda-setter asks each agent $i \in N$ to reveal his/her value for the public good and simultaneously makes a take-it-or-leave-it policy proposal consisting of a public-good level $g(\theta)$ and a vector of transfers to finance it, $t_i(\theta)$.
    \item The voting stage occurs secretly. If at least $q \leq n$ agents vote in favor of the policy proposal, then it is implemented. If not, an outside option public-good level is implemented, yielding utilities $\Bar{v}_i(\theta_i, g^\circ)$ to each agent $i \in N$.
\end{enumerate}
\bigskip
Denote
\begin{equation*}
    \mathcal{Q} \equiv \left\{Q \subseteq N: a \in Q \; \text{and} \; |Q| = q \right\}
\end{equation*}
as the set of all minimal winning coalitions that include $a$.
The agenda-setter solves the following problem
\begin{equation*}
\begin{gathered}
\max_{g(\theta), t(\theta), Q} \qquad \int \theta_a \phi(g(\theta)) -t_a(\theta) \ dF(\theta) \\
\text{subject to} \qquad \; g(\theta) \leq t_a(\theta) + \sum_{i \in N \setminus \{a\}}t_i(\theta),\\
\theta_i \phi(g(\theta)) -t_i(\theta) \geq \theta_i \phi(g(\theta'_i, \theta_{-i})) - t_i(\theta'_i, \theta_{-i}) \qquad \forall \theta_i, \theta'_i \in \Theta_i, \theta_{-i} \in \Theta_{-i} \ \text{and} \ i \in N \setminus \{a\} , \\
\text{and for some $Q \in \mathcal{Q}$}, \\
\theta_j \phi(g(\theta)) -t_j(\theta) \geq \Bar{v}_j(\theta_j, g^\circ) \qquad \forall \; j \in Q.
\end{gathered}
\end{equation*}
The first constraint is the \textit{resource constraint}: the net contributions of the agents have to be enough to cover the public-good expenditure.\\
The second is the \textit{incentive compatibility constraint}: every agent $i \in N \setminus \{a\}$ needs to have an incentive to reveal his/her true type truthfully to the agenda-setter.\\
The last constraint is the \textit{participation constraint}: the utility yielded by the proposed policy to the agents in the coalition needs to be at least as high as the one yielded by the outside option.

The participation constraint needs to be satisfied only for those $q-1$ agents included in the coalition; the remaining agents can be forced to participate. However, they have to be given the incentive to reveal their type truthfully; hence, the incentive compatibility constraint has to be satisfied for all agents.

This is consistent with alternative interpretations of the model. First, there might be political economy constraints preventing the agenda-setter to extract resources from agents who do not participate in the mechanisms.\footnote{By this, I do not mean the agents who participate but whose participation constraints are violated. Rather, they are agents who are not even asked by the agenda-setter to participate in the mechanism; e.g., this would be the case if the agenda-setter chose randomly $q-1$ agents and proposed them a direct mechanism.} Second, the agenda-setter might be thought of as a government placing any set of welfare weights on the agents' utilities. For example, an egalitarian government might place a unit weight on the agent(s) with the lowest type. In order to identify the type(s) of interest, the types of all the agents have to be elicited.\footnote{However, when the agenda-setter is indeed one of the agents seeking to maximize his/her expected utility, it might be reasonable to assume that the incentive compatibility constraint has to hold only for the members of the winning coalition. In \Cref{sec:relaxss}, it is shown that the agents in the winning coalition $Q$ will be chosen randomly.}

Finally, the incentive compatibility constraint is not imposed on the agenda-setter. Dominant strategy implementation allows for different modelling choices. First, it might be assumed that the agenda-setter's value for the public good is common knowledge. Alternatively, it might be assumed that the agenda-setter is also bound to report his/her type.\footnote{There are two possibilities. In dominant strategies, the agenda-setter benefits from misreporting and s/he can commit to a reporting strategies which does not affect the other agents' incentive compatibility constraints. If implementation is relaxed to Bayesian strategies, then the agenda-setter is an informed principal; it can be shown that s/he benefits from not revealing his/her true type. Section \ref{sec:informed} elaborates on this.}

To begin with, a few assumptions are in order. They will be relaxed in \Cref{sec:relaxss}.
\begin{assumption} \label{assmpt:fc}
The agenda-setter has full commitment power.
\end{assumption}
The agenda-setter commits to the mechanism after learning his/her type, but before eliciting the other agents' values for the public good.
This assumption rules out the possibility that the agenda-setter might modify the allocation rule after it has been announced and the preferences of the other agents have been revealed.
\begin{assumption}
\label{assumpt:regular}
The density functions $f_i(\theta_i)$, for every $i \in N$, are log-concave.
\end{assumption}
This assumption is standard in the literature and it ensures that the inverse of the hazard rate is decreasing in types.
\begin{assumption}
\label{assumpt:homogen}
    The reservation utility profile is implementable.
\end{assumption}
This assumption is referred to as \textit{homogeneity} in \cite{jullien2000}. It requires that the reservation utility profile be implementable without excluding any agent. It implies that $\Bar{v_i}$ is non-decreasing in types. 
\begin{assumption} \label{assumpt:neutral}
    For any permutation $\sigma : N \setminus \{a\} \longrightarrow N \setminus \{a\}$, every profile $\theta \in \Theta$ and every $i \in N \setminus \{a\}$, it holds that
    \begin{equation*}
        t_i(\theta_{\sigma(1)}, \dots, \theta_{\sigma(n)})=t_{\sigma(i)}(\theta_1, \dots, \theta_n) \ .
    \end{equation*}
\end{assumption}
This assumption rules out the possibility for the agenda-setter to discriminate among agents; e.g., pick an individual at random and make him/her bear the full cost of the proposal. In other words, only the reported types matter, and not the labelling (or numbering) of the agents.
\medskip\\
Standard results in mechanism design theory (\cite{maggi1995}) yield the following lemma.
\begin{lemma} \label{lemma:ic}
A direct mechanism $(g, t_i)_{i\in N}$ is dominant-strategy incentive
compatible if and only if, for every $i \in N\setminus\{a\}$ and every
$\theta_{-i} \in \Theta_{-i}$:
\begin{enumerate}[label=(\roman*)]
\item $g(\theta_i, \theta_{-i})$ is non-decreasing in $\theta_i$;
\item for every $\theta_i \in [\ubar{\theta}, \bar{\theta}]$, it holds that
\begin{equation} \label{eq:envelope}
\frac{\partial U_i}{\partial \theta_i}
= \phi(g(\theta)) - \Bar{v}_i'(\theta_i, g^\circ)
\qquad \text{almost everywhere,}
\end{equation}
where
$U_i(\theta_i, \theta_{-i}) \equiv \theta_i \phi(g(\theta_i, \theta_{-i}))
- t_i(\theta_i, \theta_{-i}) - \Bar{v}_i(\theta_i, g^\circ)$.
\end{enumerate}
\end{lemma}
The first condition is the standard \textit{monotonicity constraint}, a necessary condition for implementability of the allocation rule.

Condition $(\ref{eq:envelope})$ is an envelope condition which shows the two `countervailing incentives' at play. On the one hand, agents have the incentive to understate their value for the public-good: since it is nonrival and nonexcludable, all agents will benefit from its consumption, independently of who contributes to its financing.

On the other hand, agents have the incentive to overstate their types so as to overstate their reservation utility, the utility yielded by the outside option. Depending on which effects dominates, agents will have the overall incentive to understate or overstate their value for the public good.

Condition $(\ref{eq:envelope})$ implies that, over the set of types for whom the dominating incentive is to understate, information rents will be \textit{increasing} in types. Analogously, when the dominating incentive is to overstate the types, information rents will be \textit{decreasing} over this set of types.

This implication of dominant-strategy incentive compatibility will play a major role in the subsequent analysis.

\section{Linear Reservation Utility Profile} \label{sec:linear}
In order to foreshadow the main results, we consider a particular functional form for the reservation utility profiles, namely the case in which $\Bar{v}_i(\cdot)$ is linear in types:
\begin{equation*}
 \Bar{v}_i(\theta_i, g^\circ)=\theta_i \phi(g^\circ)-\frac{g^\circ}{n}, \qquad  \forall i \in N.   
\end{equation*}
The outside option is given by a public-good level $g^\circ$, and a uniform array of transfers to finance it.\footnote{This is a common assumption in the political economy literature; see e.g., \cite{battaglini2008} and \cite{pitsuwan2023}.} Importantly, by condition (\ref{eq:envelope}), the incentive to understate or overstate the type is given by
\begin{equation} \label{eq:linearic}
        dU_i/d\theta_i=\phi(g(\theta))-\phi(g^\circ)  \ .
    \end{equation}
Hence, the sign of (\ref{eq:linearic}) is type-independent. The analysis is thus greatly simplified.

\subsection{Unanimity}
Suppose first that the institutional constraints require that the proposal be approved by unanimity; in other words, $q=n$.
The agenda-setter solves the following problem\\
\begin{equation*} 
\begin{gathered} \tag{$\mathcal{LU}$} \label{problem:lu}
\max_{g(\theta), \bold{t}(\theta)} \qquad \int \theta_a \phi(g(\theta)) -t_a(\theta) \ dF(\theta) \\
\text{subject to} \qquad \; g(\theta) \leq  t_a(\theta) + \sum_{i \in N\setminus \{a\}}t_i(\theta),\\
g(\theta_i, \theta_{-i}) \ \text{increasing in} \ \theta_i  , \\
dU_i/d\theta_i=\phi(g(\theta))-\phi(g^\circ), \\
\theta_i \phi(g(\theta)) -t_i(\theta) \geq \theta_i \phi(g^\circ)-\frac{g^\circ}{n} \qquad \forall \; i \in N.
\end{gathered}
\end{equation*}
\medskip\\
Denote $g^*(\theta)$ as the solution to \ref{problem:lu}.
\begin{lemma} \label{lemma:ir}
   Let $(g, t_i)_{i \in N}$ be dominant-strategy incentive compatible, and fix
$i \in N \setminus \{a\}$ and $\theta_{-i} \in \Theta_{-i}$. Then, if $g(\theta_i, \theta_{-i}) \neq g^\circ$, then the participation constraint binds at most at one type.
\end{lemma}

Then, from \Cref{lemma:ic} and \Cref{lemma:ir} we immediately have the following result.
\begin{proposition} \label{prop:unanimtd}
    There exist two threshold values, $g_U^L(\theta)$ and $g_U^H(\theta)$, such that:
    \begin{equation*}
g^*(\theta)=\begin{cases}
g_U^L(\theta)&\mbox{if $g^\circ < g_U^L(\theta)$,}\\
g^\circ &\mbox{if $g_U^L(\theta)\leq g^\circ \leq g_U^H(\theta)$,}\\
g_U^H(\theta)&\mbox{if $g^\circ > g_U^H(\theta)$.}
\end{cases}
\end{equation*}
where $g_U^L(\theta)$ and $g_U^H(\theta)$ solve, respectively,
\begin{align*}
    \left[\theta_a + \sum_{i \in N\setminus \{a\}} \left(\theta_i - \frac{1-F_i(\theta_i)}{f_i(\theta_i)} \right)\right]\phi'(g(\theta))&=1 \\
     \left[\theta_a + \sum_{i \in N\setminus \{a\}} \left(\theta_i + \frac{F_i(\theta_i)}{f_i(\theta_i)} \right)\right]\phi'(g(\theta))&=1 \ .
\end{align*}
\end{proposition}
\Cref{fig:linearunanimity} illustrates the solution.
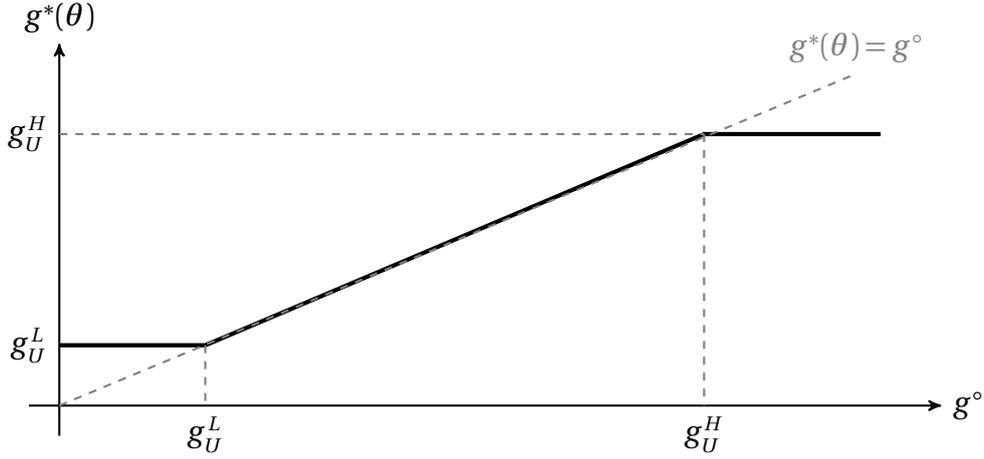
\begin{figure}[ht]
\centering
    \begin{tikzpicture}[-,>=stealth',auto,node distance=1.2cm,
  thick,main node/.style={transform shape, circle,draw,font=\rmfamily, minimum width=5pt}, scale=0.8]

    \draw[->] (-5,0) -- (10,0) node[right]{$g^\circ$};
    \draw[->] (-4.5,-0.5) -- (-4.5,6) node[above]{$g^*(\theta)$};

  \draw[ultra thick] (-4.5,1)node[left]{$g_U^L(\theta)$} -- (-2.1,1) -- (6.1,4.5) -- (9,4.5);
   \draw (-2.1,0)node[below]{$g_U^L(\theta)$};
   \draw (6.1,0)node[below]{$g_U^H(\theta)$};
   \draw (-4.5,4.5)node[left]{$g_U^H(\theta)$};
    \draw[dashed, color=gray] (-4.5,0) -- (8.6,5.5) node[above]{$g^*(\theta) = g^\circ$}; 
    \draw[dashed, color=gray] (-2.1,1)--(-2.1,0);
    \draw[dashed, color=gray] (6.1,4.5)--(6.1,0);
    \draw[dashed, color=gray] (-4.5,4.5)--(6.1,4.5);
\end{tikzpicture}
\caption{$g^*(\theta)$ in the linear case with unanimity.}
\label{fig:linearunanimity}
\end{figure}\\
The intuition behind the result is very simple. Since all participation constraints need to be satisfied, the incentive to understate or overstate the type is \textit{ex post} type-independent: once a public-good level $g^*(\theta)$ has been provided, all agents will have the \textit{ex post} incentive to either understate or overstate their type, depending on the sign of (\ref{eq:linearic}). 

If the outside option public-good level is sufficiently low $(g^\circ \leq g_U^L(\theta))$, all agents will have the incentive to understate their type. If the outside option public-good level is sufficiently high $(g^\circ \geq g_U^H(\theta))$, all agents will have the incentive to overstate their type.

Finally, if the outside option public-good level assumes intermediate values $(g_U^L(\theta) \leq g^\circ \leq g_U^H(\theta))$, then some agents will have the incentive to understate, whereas some other agents will have the incentive to overstate. The two conflicting incentives cannot be consistent \textit{ex post} with (\ref{eq:linearic}). Hence, the agenda-setter can only propose the outside-option public-good level.

\subsection{Majority}
Suppose now that $q \leq n$.
The agenda-setter solves
\begin{equation} \tag{$\mathcal{LM}$} \label{problem:lm}
\begin{gathered}
\max_{g(\theta), \bold{t}(\theta), Q} \qquad \int \theta_a \phi(g(\theta)) -t_a(\theta) \ dF(\theta)\\
\text{subject to} \qquad \; g(\theta) \leq t_a(\theta) + \sum_{i \in N \setminus \{a\}}t_i(\theta),\\
g(\theta_i, \theta_{-i}) \ \text{increasing in} \ \theta_i ,\\
\partial U_i/ \partial \theta_i=\phi(g(\theta))-\phi(g^\circ) \qquad \forall \; i \in N \setminus \{a\}, \\
\text{and for some $Q \in \mathcal{Q}$}, \\
\theta_j \phi(g(\theta)) -t_j(\theta) \geq \theta_j \phi(g^\circ)-\frac{g^\circ}{n} \qquad \forall \; j \in Q.
\end{gathered}
\end{equation}
\medskip\\
The first three constraints are the same as in the unanimity problem. However, in addition to $g(\theta)$ and $\bold{t}(\theta)$, the agenda-setter now also needs to choose a coalition $Q \in \mathcal{Q}$.

\begin{figure}[ht]
\centering
    \begin{tikzpicture}[-,>=stealth',auto,node distance=1.2cm,
  thick,main node/.style={transform shape, circle,draw,font=\rmfamily, minimum width=5pt}, scale=0.8]

    \draw[->] (-5,0) -- (10,0) node[right]{$g^\circ$};
    \draw[->] (-4.5,-0.5) -- (-4.5,6) node[above]{$g^*(\theta)$};

  \draw[ultra thick] (-4.5,1.5)-- (-1,1.5) -- (5,4) -- (9,4);
   \draw(-4.5,1.5)node[left]{$g_M^L(\theta)$};
   \draw (-1,0)node[below]{$g_M^L(\theta)$};
   \draw (5,0)node[below]{$g_M^H(\theta)$};
   \draw (-4.5,4.5)node[left]{$g_M^H(\theta)$};
    \draw[dashed, color=gray] (-4.5,0) -- (8.6,5.5) node[above]{$g^*(\theta) = g^\circ$}; 
    \draw[dashed, color=gray] (-1,1.5)--(-1,0);
    \draw[dashed, color=gray] (5,4)--(5,0);
    \draw[dashed, color=gray] (-4.5,4)--(5,4);
\end{tikzpicture}
\caption{$g^*(\theta)$ in the linear case with majority ($g_M^L(\theta)<g_M^H(\theta)$).}
\label{fig:linearmajority-1}
\end{figure}

\bigskip

For any $Q \in \mathcal{Q}$, let $\Xi(Q) = \sum_{i \in Q} \theta_i$ be $Q$'s value for the public good. Define
\begin{equation*} 
\underline{Q} \equiv \argmin_{Q \in \mathcal{Q}} \Xi(Q) \quad \text{and} \quad \overline{Q} \equiv \argmax_{Q \in \mathcal{Q}} \Xi(Q),
\end{equation*}
to be, respectively, the `lowest' and `highest' coalitions as measured by their value for the public good.
Along the same lines as in the unanimity case, we obtain the following

\begin{proposition} \label{prop:majoritd}
There exist a set of agents $I$, with $|I| \leq n-q$, and two threshold public-good values, $g_M^L(\theta)$, $g_M^H(\theta) \in \mathbb{R}_+$, such that, if $g_M^L(\theta) < g_M^H(\theta)$, 
\begin{equation*}
g^*(\theta)=\begin{cases}
g_M^L(\theta)&\mbox{if $g^\circ < g_M^L(\theta)$,}\\
g^\circ &\mbox{if $g_M^L(\theta)\leq g^\circ \leq g_M^H(\theta)$,}\\
g_M^H(\theta)&\mbox{if $g^\circ > g_M^H(\theta)$,}
\end{cases}
\end{equation*}
where $g_M^L(\theta)$ and $g_M^H(\theta)$ solve, respectively, 
\begin{align*}
    \left[ \theta_a + \sum_{i \in I} \left(\theta_i + \frac{F_i(\theta_i)}{f_i(\theta_i)} \right) + \sum_{i \in \overline{Q}\setminus \{a\}} \left(\theta_i - \frac{1-F_i(\theta_i)}{f_i(\theta_i)} \right)\right]\phi'(g(\theta)) &= 1 \ ,\\
   \left[ \theta_a + \sum_{i \in I} \left(\theta_i - \frac{1-F_i(\theta_i)}{f_i(\theta_i)} \right) + \sum_{i \in \underline{Q}\setminus \{a\}} \left(\theta_i + \frac{F_i(\theta_i)}{f_i(\theta_i)} \right)\right]\phi'(g(\theta)) &= 1 \ .
\end{align*}
Moreover,
\begin{equation*}
Q^*=\begin{cases}
\overline{Q}&\mbox{if $g^\circ < g_M^L(\theta)$,}\\
\textrm{any}  \ Q \in \mathcal{Q} &\mbox{if $g_M^L(\theta)\leq g^\circ \leq g_M^H(\theta)$,}\\
\underline{Q}&\mbox{if $g^\circ > g_M^H(\theta)$.}
\end{cases}
\end{equation*}
\bigskip\\
If $g_M^L(\theta)>g_M^H(\theta)$ for some $q\leq n$, then there exists an outside option $\hat{g}^\circ$ such that
\begin{equation*}
g^*(\theta)=\begin{cases}
g_M^L(\theta)&\mbox{if $g^\circ < \hat{g}^\circ$,}\\
g_M^H(\theta)&\mbox{if $g^\circ > \hat{g}^\circ$.}
\end{cases}
\end{equation*}
Moreover,
\begin{equation*}
Q^*=\begin{cases}
\overline{Q}&\mbox{if $g^\circ < \hat{g}^\circ$,}\\
\underline{Q}&\mbox{if $g^\circ > \hat{g}^\circ$}.
\end{cases}
\end{equation*}
\end{proposition}
\bigskip
\Cref{fig:linearmajority-1} and \Cref{fig:linearmajority-2} illustrate the solution.

\bigskip

    \begin{figure}[ht]
\centering
    \begin{tikzpicture}[-,>=stealth',auto,node distance=1.2cm,
  thick,main node/.style={transform shape, circle,draw,font=\rmfamily, minimum width=5pt}, scale=0.8]

    \draw[->] (-5,0) -- (10,0) node[right]{$g^\circ$};
    \draw[->] (-4.5,-0.5) -- (-4.5,6) node[above]{$g^*(\theta)$};

  \draw[ultra thick] (-4.5,4) -- (2.5,4);
  \draw[ultra thick] (2.5,1.5) -- (8,1.5);
   \draw (-4.5,1.5)node[left]{$g_M^H$};
   \draw (-1,0)node[below]{$g_M^H$};
   \draw (5,0)node[below]{$g_M^L$};
   \draw (-4.5,4)node[left]{$g_M^L$};
   \draw[color=red] (2.5,0)node[below]{$\hat{g}$};
    \draw[dashed, color=gray] (-4.5,0) -- (8.6,5.5) node[above]{$g^*(\theta) = g^\circ$}; 
    \draw[dashed, color=gray] (-4.5,1.5)--(8,1.5);
    \draw[dashed, color=gray] (5,4)--(5,0);
    \draw[dashed, color=gray] (-0.9,1.5) -- (-0.9,0);
    \draw[dashed, color=gray] (-4.5,4)--(5,4);

     \draw[dashed, color=red] (2.5,4)--(2.5,0);
\end{tikzpicture}
\caption{$g^*(\theta)$ in the linear case with majority ($g_M^L(\theta)>g_M^H(\theta)$).}
\label{fig:linearmajority-2}
\end{figure}

Importantly, the participation constraints of some agents can be violated; however, they still need to be given incentives to truthfully report their types. 

Unlike the unanimity case, the agenda-setter can now exclude some agents from the winning coalition; i.e., some agents can be forced to participate in the collective choice. However, they have to be given the incentive to truthfully report their types. Therefore, they will also affect public-good provision.

Interestingly, the majority rule creates \textit{endogenously} countervailing incentives; by excluding agents from the winning coalition, the agenda-setter is able to extract more resources and increase (or decrease) the public-good provision. The intuition behind \Cref{prop:majoritd} is again related to (\ref{eq:linearic}).

If $g^*(\theta) > g^\circ$, information rents have to be increasing in types; however, the participation constraint now binds on some interior type $\Tilde{\theta}$. Condition (\ref{eq:linearic}) then becomes:
\begin{equation*}
U_i(\theta_i) 
= U_i(\Tilde{\theta}) + \int_{\Tilde{\theta}}^{\theta_i}
\bigl[\phi(g(x, \theta_{-i})) - \phi(g^\circ)\bigr]\,dx \ .
\end{equation*}

For all types above $\Tilde{\theta}$, incentive constraints bind downwards, and agents must earn information rents as in the unanimity case. However, all types below have their participation constraints violated. The agenda-setter pays them \textit{negative} information rents; in other words, s/he extracts resources from them. Since $U_i(\cdot)$ is convex, information rents for the excluded types are \textit{concave}; the agenda-setter extracts more resources from lower types. Crucially, this allows the agenda-setter to increase the public-good level from the inefficiently low level of the unanimity case. Therefore, the incentive constraints of the excluded types bind \textit{upwards}.

Figure \ref{fig:linrents} illustrates these dynamics. An analogous reasoning goes through for the case where $g^\circ$ is high.

\begin{figure}
    \centering
    \begin{tikzpicture}[scale=1.0]
\draw[->] (0,0) -- (9.3,0) node[right] {$\theta_i$};
\draw[->] (0,-1.8) -- (0,2.4) node[above] {$U_i$};
\draw[dashed] (4.5,-1.8) -- (4.5,2.2);
\node[below] at (4.5,-1.9) {$\ct$};
\node[below] at (0.15,-0.05) {$\ubar\theta$};
\node[below] at (9.0,-0.05) {$\bar\theta$};
\draw[dotted] (0,0) -- (9.0,0);
\draw[very thick, blue] (0.3,-1.5) .. controls (2.5,-0.9) and (3.8,-0.35) ..
(4.5,0);
\draw[very thick, red] (4.5,0) .. controls (5.5,0.55) and (7.2,1.4) ..
(9.0,2.1);
\draw[->,thick,blue] (2.0,-1.35) -- (3.0,-1.35);
\node[below,blue] at (2.5,-1.4) {\small overstate};
\draw[->,thick,red] (8.0,1.75) -- (7.0,1.75);
\node[above,red] at (7.5,1.8) {\small understate};
\node[blue] at (2.1,-0.45) {\small extraction region};
\node[red] at (7.2,0.5) {\small rent region};
\end{tikzpicture}
    \caption{Endogenous countervailing incentives}
    \label{fig:linrents}
\end{figure}

Interestingly, when for some $q$ it is the case that $g_M^L(\theta) > g_M^H(\theta)$, public-good provision becomes non-monotonic in the outside option public-good level. This is the same result obtained in \cite{pitsuwan2023}.

The results in this section come as no surprise. When the reservation utility profile is linear in types, condition (\ref{eq:linearic}) states that the agents' incentive to misreport their types is type-independent. Hence, the agenda-setter is able to implement a collective choice which is different from the outside option only when the incentives are aligned for all agents. This is the case when the outside option public-good level is sufficiently low or sufficiently high. 

On the other hand, when the outside option public-good level assumes intermediate values, the only collective choice which can be implemented in dominant strategies is proposing the outside option.

A similar dynamic unfolds in the majority case. However, by violating the participation constraints of some agents, the agenda-setter creates countervailing incentives, thus increasing (or decreasing) public-good provision to a level closer to the first-best.

\section{Main Results} \label{sec:main}
In this section, I consider a general reservation utility profile $\Bar{v}_i$. I am going to assume that the reservation utility profile satisfies a \textit{single-crossing} condition.

\begin{assumption} \label{ass:singlecrossing}
   $ \frac{\partial^2 \Bar{v}_i}{\partial \theta_i \partial g^\circ} >0$ \ .
\end{assumption}
Assumption \ref{ass:singlecrossing} implies that higher types have higher marginal utilities under the outside option. It will be relaxed in Section \ref{sec:negslope}.

The agenda-setter solves
\begin{equation} \tag{$\mathcal{G}$} \label{gproblem}
\begin{gathered}
\max_{g(\theta), \bold{t}(\theta)} \qquad \int \theta_a \phi(g(\theta)) -t_a(\theta) \ dF(\theta)\\
\text{subject to} \qquad \; g(\theta) \leq  t_a(\theta) + \sum_{i \in N\setminus \{a\}}t_i(\theta),\\
g(\theta_i, \theta_{-i}) \ \text{increasing in} \ \theta_i ,\\
\partial U_i/ \partial \theta_i=\phi(g(\theta))-\Bar{v}_i'(\theta_i, g^\circ) \qquad \forall \; i \in N \setminus \{a\}, \\
\text{and for some $Q \in \mathcal{Q}$}, \\
\theta_j \phi(g(\theta)) -t_j(\theta) \geq \Bar{v}_j(\theta_j, g^\circ) \qquad \forall \; j \in Q.
\end{gathered}
\end{equation}

Theorem \ref{thrm:u} characterizes the set of binding constraints and the public-good level of the optimal mechanism for each outside option public-good level. Then, Theorem \ref{thrm:m} identifies the optimal coalition, as well as the set of agents whose participation constraint is violated. Before stating the two main results of this paper, some preliminary definitions and results are in order.

A key insight from \cite{jullien2000} is that the shadow value on the participation constraints at the optimal allocation can be represented by $\gamma_i^*(\theta_i)$, which is a cumulative distribution function over the set of types. The agenda-setter surplus can be rewritten so as to account for the information rents brought about by incentive compatibility.
\begin{lemma} \label{lemma:objfct}
    The agenda-setter's objective function can be rewritten as
    \begin{align*} \label{eq:objfct}
   \sigma(\theta, \gamma(\theta), g(\theta); g^\circ) &\equiv   \int_{\Theta} \left[\sum_{i \in N} \theta_i\phi(g(\theta)) - \sum_{i \neq a} \left(\frac{\gamma^*(\theta_i)-F_i(\theta_i)}{f_i(\theta_i)} \right)\phi(g(\theta)) - g(\theta)\right] f(\theta) d\theta + \\  
   & + \int_{\Theta} \sum_{i \neq a} \Bar{v}_i(\theta_i, g^\circ) \ d\gamma_i(\theta_i)\ .
    \end{align*}
\end{lemma}
Since the objective function is strictly quasi-concave in the public-good level, the objective function will have a unique maximizer.
\begin{definition}
    \begin{equation*}
        \xi(\gamma(\theta), \theta) = \argmax_{g(\theta)} \ \sigma(\theta, \gamma(\theta), g(\theta); g^\circ).
    \end{equation*}
\end{definition}

We are now ready to state the main results of the paper.
\begin{theorem} \label{thrm:u}
Under Assumptions (\ref{assmpt:fc})-(\ref{ass:singlecrossing}), there exists a unique optimal solution to \ref{gproblem}. An implementable allocation $g^*(\theta)$ is optimal if and only if there exists a cumulative distribution function  $\gamma^*(\theta_i)$ such that
\begin{equation*}
\begin{cases}
&(\mathcal{A}) \qquad \operatorname{supp}  d\gamma_i^*(\theta) \subseteq \{\theta_i \in Q(\theta) \ | \ U_i(\theta_i)=0 \} \ ,\\
&(\mathcal{B}) \qquad g^*(\theta)= \xi(\gamma^*(\theta_i), \theta).
\end{cases}
\end{equation*}
\end{theorem}
\Cref{thrm:u} is closely related to Theorem 1 in \cite{jullien2000}. Condition $\mathcal{A}$ identifies which participation constraints bind and, hence, which types earn no information rents. It states that the participation constraint can bind on a single interior type and at the extremes: the lowest and highest types.

Condition $\mathcal{B}$ characterizes the optimal public-good provision as the public-good level which maximizes the agenda-setter’s utility adjusted to take into account the information rents induced by incentive compatibility. Depending on the outside option public-good level $g^\circ$, some agents will have the incentive to understate their type, whereas others will have the incentive to overstate it. In particular, if $\gamma^*(\theta_i) > F_i(\theta_i)$, the dominating incentive is to understate the type; if $\gamma^*(\theta_i) < F_i(\theta_i)$, the dominating incentive is to overstate the type.

Unlike the linear case, countervailing incentives might arise even under unanimity for specific functional forms of the reservation utility profile.

In order to characterize the optimal coalition, we need the following notation.
\begin{definition} \label{def:expostrent}
Denote
\begin{equation*}
    J(\theta_i)= \int_{\Tilde{\theta}}^{\theta} \phi(g(x, \theta_{-i}))- \Bar{v}_i'(x, g^\circ) \ dx
\end{equation*}
the \textit{ex post rent}.
\end{definition}
Now, define a partition of the agents according to the realized rents
\begin{equation*}
    R= \{ \theta_{i_1}, \theta_{i_2}, \dots, \theta_{i_{N-1}} \} \qquad \text{with} \quad J(\theta_{i_1}) \leq  J(\theta_{i_1}) \leq \dots \leq J(\theta_{i_{N-1}}) \ .
\end{equation*}

\begin{theorem} \label{thrm:m}
The solution to Problem \ref{gproblem} induces a partition of the set of agents such that
\begin{equation*}
    Q^*=\{a\} \cup \left\{\theta_i \in R \ | \ \theta_i \leq \theta_{i_{q-1}} \right\} 
\end{equation*}
and 
\begin{equation*}
     I^*(g^\circ) = \left\{ \theta_i \in R \ | \  \theta_i > \theta_{i_{q-1}} \right\} \ ,
\end{equation*}
with $U_i(\theta_i) <0$ for all $i \in I(g^\circ)$. Moreover, the participation constraints bind on $\Tilde{\theta}=\theta_{i_1}$. 
\end{theorem}

The interpretation of \Cref{thrm:m} goes along the following lines. Intuitively, the agenda-setter finds it optimal to satisfy not more than $q-1$ participation constraints; this is because they can be extracted resources from, with which to pay information rents to the members of the coalition $Q$. However, \textit{ex post} incentive compatibility places strong restrictions as to which agents can be excluded from the coalition.

In fact, by the envelope condition (\ref{eq:envelope}), information rents have to be increasing (decreasing) over the set of agents who have the incentive to understate (overstate) their types. Importantly, the agenda-setter has to identify a set of agents whose participation constraints can be violated, but whose incentives to report truthfully are still required to hold. The choice of this set crucially depends on the shape of the reservation utility profile.

Notice that this type of exclusion is different from the exclusion of some agents in \cite{jullien2000}. The reason is that, in the present framework, agents can be forced to participate in the mechanism, and they will still consume the public good once it has been provided. However, they still need to be given the incentives to reveal their types truthfully \textit{ex post}.

Therefore, \Cref{thrm:u} and \Cref{thrm:m} imply that the optimal public-good provision is either equal to - or a step function of - the outside option public-good level.

\begin{figure}[ht]
\centering
    \begin{tikzpicture}[-,>=stealth',auto,node distance=1.2cm,
  thick,main node/.style={transform shape, circle,draw,font=\rmfamily, minimum width=5pt}]

    \draw[->] (-5,0) -- (10,0) node[right]{$g^\circ$};
    \draw[->] (-4.5,-0.5) -- (-4.5,6) node[above]{$g^*(\theta)$};

  \draw[ultra thick] (-4.5,1)node[left]{$g^L$} -- (-2.1,1);
  \draw[ultra thick] (-2.1,2) -- (0.2,2); 
  \draw[ultra thick] (0.2,3)-- (2.6,3);
  \draw[ultra thick] (2.6,4) -- (5,4);
   \draw[ultra thick] (5,4.8) -- (9,4.8);
   \draw (-2.1,0)node[below]{$g^L$};
   \draw (7,0)node[below]{$g^H$};
   \draw (-4.5,4.8)node[left]{$g^H$};
    \draw[dashed, color=gray] (-4.5,0) -- (8.6,5.5) node[above]{$g^*(\theta) = g^\circ$}; 
    \draw[dashed, color=gray] (-2.1,1)--(-2.1,0);
    \draw[dashed, color=gray] (7,4.8)--(7,0);
    \draw[dashed, color=gray] (-4.5,4.8)--(7,4.8);
\end{tikzpicture}
\caption{$g^*(\theta)$ is a step function of $g^\circ$.}
\label{fig:general}
\end{figure}
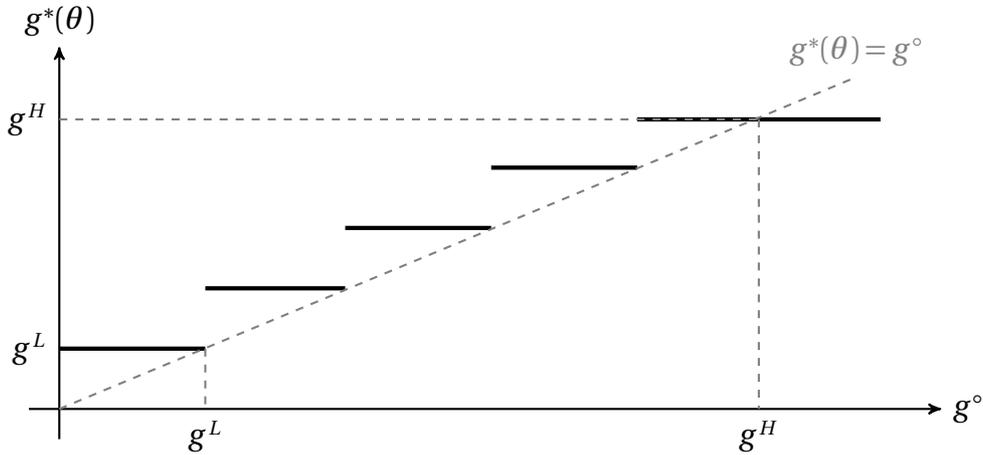

\section{Implications} \label{sec:implications}
\Cref{thrm:u} and \Cref{thrm:m} show that optimal public-good provision crucially depends on the shape of the reservation utility profile $\Bar{v}_i$. 

Consider the unanimity case.
When the outside option public-good level $g^\circ$ is sufficiently low, all agents have the incentive to understate their types. By familiar arguments, it is then optimal to underprovide the public good by distorting the allocations downwards. The participation constraint binds for the lowest type; all others earn information rents. By the envelope condition (\ref{eq:envelope}), rents are increasing in types.

Analogously, when the outside option public-good level is sufficiently high, all agents have the incentive to overstate their types. The optimal public-good provision now entails distorting the allocation upwards. The participation constraints binds only for the highest type. All other agents earn information rents, which are now decreasing in types by (\ref{eq:envelope}).

When $g^\circ$ assumes intermediate values, however, agents have countervailing incentives \textit{even under unanimity}: some agents will have the incentive to understate their types, whereas others will have the incentive to overstate their types. This is due to the reservation utility profile, and not to the incentives of the agenda-setter.

Therefore, the optimal public-good provision depends on the concavity or convexity of the reservation utility profile.

\Cref{sec:linear} considered the linear case. We now consider the concave and convex cases.

In order to simplify the exposition, it is useful to define a partition of the agents included in the coalition in terms of their \textit{ex post} incentive to understate or overstate their types.
\begin{definition}
\begin{equation*}
  N= \{a\} \ \cup \ K^Q(g^\circ) \ \cup \ L^Q(g^\circ) \ \cup \ M^Q(g^\circ) \ \cup \ I(g^\circ) \ ,
\end{equation*}
where
\begin{align*}
    K^Q(g^\circ) &= \left\{ \theta_i \in Q \ | \ \phi(g(\theta_i, \theta_{-i})) < \Bar{v_i}(\theta_i, g^\circ) \right\} \\
    L^Q(g^\circ) &= \left\{ \theta_i \in Q  \ | \  \phi(g(\theta_i, \theta_{-i})) = \Bar{v_i}(\theta_i, g^\circ) \right\} \\
    M^Q(g^\circ) &= \left\{ \theta_i \in Q  \ | \   \phi(g(\theta_i, \theta_{-i})) > \Bar{v_i}(\theta_i, g^\circ) \right\} \\ 
    I(g^\circ) &= \left\{ \theta_i \in N \setminus Q \ | \ U_i(\theta_i) <0 \right\} \ .
\end{align*}
\end{definition}

Consider now the implicit function
\begin{equation} \label{eq:G}
    G(\theta_i; g^\circ) \equiv \phi(g(\theta)) - \Bar{v}'_i(\theta_i, g^\circ) =0 \ . 
\end{equation}
The following result ensures that, for given public-good provision $g(\theta)$ and outside option public-good level $g^\circ$, agents are ordered according to their incentive to misreport their types. A crucial role is played by the shape of the reservation utility profile.

 \begin{lemma} \label{lemma:propofG} For $\theta_i \in [\underline{\theta}, \overline{\theta}]$ and $z \in \mathbb{R}_+$, there exists a unique $\bar{g}(\theta_i, z)$ such that \[G(\theta_i; \bar{g}(\theta_i, z)) = z.\] 
In addition,  $\bar g$ is \emph{strictly decreasing in $z$}, and, if
$\vbar_i$ is concave (convex) in $\theta_i$, strictly increasing
(decreasing) in $\theta_i$.
 \end{lemma}

As \Cref{lemma:propofG} shows, the ordering of agents according to their types hinges on the concavity or convexity of the reservation utility profile. Therefore, each case is considered separately.

\subsection{Concave Reservation Utility Profile}
\begin{lemma} \label{lemma:partition}
Suppose $\bar{v}''_i \leq 0$. Then,  for some $x,y \in \{0,1,\dots,r\}$ with $y=x+n-(q-1)$,
\begin{align*}
    K^Q(g^\circ) &= \{\theta_r,\dots,\theta_{y+1}\}, \\
    L^Q(g^\circ) &= \{\theta_{x}, \theta_{y}\} \cap \{\theta_1, \dots, \theta_r \},  \\
    M^Q(g^\circ) &= \{\theta_{x-1},\dots, \theta_1\}, \\
    I(g^\circ)   &= \{\theta_{y-1}, \dots, \theta_{x+1} \} \ ,
\end{align*}
with the convention that $\{\theta_c,\dots,\theta_d\}=\varnothing$ if $c<d$.
Moreover, the indices $x$ and $y$ are increasing in $g^\circ$.
\end{lemma}
\Cref{lemma:partition} orders the partition of the agents, depending on whether they have the incentive to understate or overstate their types, or they are indifferent between the two.
Intuitively, under unanimity, and When values are all distinct, there will be a single type $\Tilde{\theta}$ whose participation constraint binds. When $g^\circ$ is very low, then $\Tilde{\theta}=\theta_1$. As the outside option public-good level increases, the cutoff type $\Tilde{\theta}$ will be an interior type. By \Cref{lemma:partition}, all types below (above) $\Tilde{\theta}$ will have the incentive to overstate (understate) their types.

\Cref{fig:rentsconcaveu} displays the information rents as a function of the types in this case.
\begin{figure}[ht]
\centering
    \begin{tikzpicture}[-,>=stealth',auto,node distance=1.2cm,
  thick,main node/.style={transform shape, circle,draw,font=\rmfamily, minimum width=5pt}]

    \draw (0,2)node[below]{$\underline{\theta}$} -- (8,2);
    \draw[->] (0,2) -- (0,6);
    \draw[->] (8,2)node[below]{$\overline{\theta}$} -- (8,6);
    
   \draw (4,2) node[below]{$\Tilde{\theta}$};
   \draw [blue, thick] (4,2) to [out=180,in=280] (0,4.5);
   \draw [red, thick] (4,2) to [out=0, in=250] (8,4.5);

\end{tikzpicture}
\caption{$\Bar{v}_i$ concave: unanimity}
\label{fig:rentsconcaveu}
\end{figure}
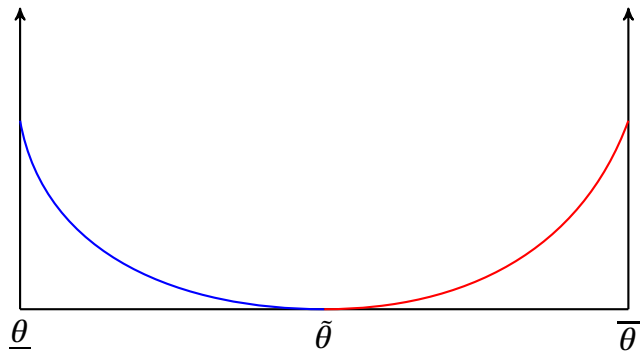
In the majority case, the same argument applies. However, \Cref{thrm:m} states that a subset of agents $I$ will have their participation constraints violated. Depending on $g^\circ$, the participation constraint binds on at most two different types and hence also the set of agents excluded from the coalition differs. \Cref{fig:rentsconcavemex} displays the solution when the cutoff types are interior.

\begin{figure}[ht]
\centering
    \begin{tikzpicture}[-,>=stealth',auto,node distance=1.2cm,
  thick,main node/.style={transform shape, circle,draw,font=\rmfamily, minimum width=5pt}]

    \draw (0,2)node[below]{$\underline{\theta}$} -- (8,2);
    \draw[->] (0,2) -- (0,6);
    \draw[->] (8,2)node[below]{$\overline{\theta}$} -- (8,6);
    
   \draw (4,1.7) node[below]{$\Tilde{\theta}$};
   \draw [blue, thick] (4,1.7) to [out=180,in=280] (0,4.2);
   \draw [red, thick] (4,1.7) to [out=0, in=250] (8,4.2);
   \filldraw [blue] (2,2) circle (1.5pt) node[below]{$\theta_p$};
   \filldraw [red] (5.8,2) circle (1.5pt) node[below]{$\theta_q$};

  \draw [<->] [orange] (0,6) -- (2,6);
  \draw [orange] (4,6)node{$Q$};
  \draw [<->] [orange]  (6,6)--(8,6);
  \draw [<->] [violet] (2,2.5) -- (4,2.5)node[above]{$I$} -- (5.8,2.5);
\end{tikzpicture}
\caption{$\Bar{v}_i$ concave with exclusion}
\label{fig:rentsconcavemex}
\end{figure}
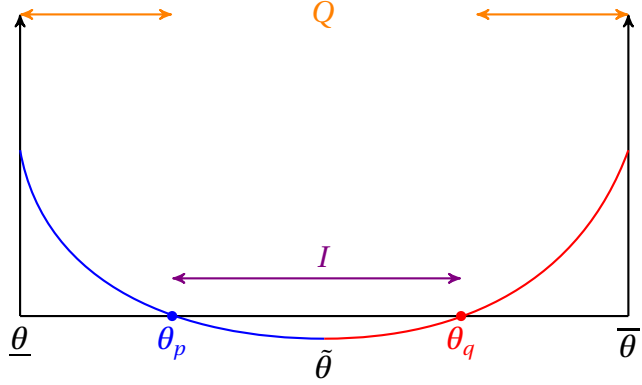
Interestingly, the winning coalition can be a \textit{non-convex} set: there can be three types $\theta_i < \theta_k < \theta_j$ such that types $\theta_i$ and $\theta_j$ might be included in $Q$, whereas $\theta_k$ is excluded.
\begin{proposition} \label{prop:concave}
    For any $g^\circ$, there exists two agents $\theta_p, \theta_q$ such that optimal public-good provision solves
    \begin{equation*}
    \left[ \theta_a + \sum_{\theta_i \in N\setminus \{a\}} \left(\theta_i - \frac{\gamma^*(\theta_i)-F_i(\theta_i)}{f_i(\theta_i)} \right) \right]\phi'(g(\theta)) = 1 \ ,
\end{equation*}
where
\begin{equation*}
    \gamma^*(\theta_i)=\begin{cases}
0 & \theta_i<\theta_p,\\
\alpha & \theta_p\le \theta_i<\theta_q,\\
1 & \theta_i\ge\theta_q \ .
\end{cases}
\end{equation*}
Moreover, the optimal coalition is given by
\begin{equation*}
 Q^*=\{a\} \cup \left\{\theta_i \in N \setminus \{a\} \ | \ \theta_i \leq \theta_{p} \ \cup \ \theta_i \geq \theta_q \right\} \ .    
\end{equation*}
\end{proposition}
\subsection{Convex Reservation Utility Profile}
\begin{lemma} \label{lemma:partition2}
Suppose $\bar{v}''_i \geq 0$. Then there exist
$x\in\{1,\dots,n-q+1\}$ and $m$, with $y=x+q-2$ and
$x\le m<y$, such that:
\begin{align*}
    K^Q(g^\circ) &= \{\theta_{r-1},\dots,\theta_l\}, \\
    L^{Q}(g^\circ)&=\{\theta_x,\ \theta_y\},  \\
    M^{Q}(g^\circ)&=\{\theta_{x+1},\dots,\theta_m\}, \\
    I(g^\circ)&=\{\theta_1,\dots,\theta_{x-1}\}
\cup\{\theta_{y+1},\dots,\theta_r\},
\end{align*}
with
$\{\theta_c,\dots,\theta_d\}=\emptyset$ if $c>d$.
Moreover, the indices $x$ and $y$ are decreasing in $g^\circ$.
\end{lemma}
Analogously to \Cref{lemma:partition}, \Cref{lemma:partition2} orders the partition of the agents. However, the transition between sets is more involved than in the concave case. The reason is again Condition (\ref{eq:envelope}).

Again, consider unanimity first. As $g^\circ$ increases, the first agent whose incentive switches from understating to overstating is the agent with the highest type. However, the information rents have to be decreasing for the set of agents having the incentive to overstate. The only way to preserve incentive compatibility is then to have the participation constraint binding at both extremes of the type space.

Interestingly, there will be a single interior type where the two countervailing incentives cancel out and the agent will be just indifferent between understating or overstating.\footnote{See \cite{maggi1995}.}
This cutoff type is decreasing in $g^\circ$: it coincides with $\theta_n$ when $g^\circ$ is very low and it coincides with lower types as $g^\circ$ increases. \Cref{fig:rentsconvexu} illustrates this property.
\begin{figure}[ht]
\centering
    \begin{tikzpicture}[-,>=stealth',auto,node distance=1.2cm,
  thick,main node/.style={transform shape, circle,draw,font=\rmfamily, minimum width=5pt}]

    \draw (0,2)node[below]{$\underline{\theta}$} -- (8,2);
    \draw[->] (0,2) -- (0,6);
    \draw[->] (8,2)node[below]{$\overline{\theta}$} -- (8,6);
    
   \draw (4,2) node[below]{$\Tilde{\theta}$};
   \draw[dotted] (4,4.5)--(4,2);
  \draw [red, thick] (4,4.5) to [out=180,in=80] (0,2);
   \draw [blue, thick] (4,4.5) to [out=0, in=100] (8,2);

\end{tikzpicture}
\caption{Information rents when $\Bar{v}_i$ is convex}
\label{fig:rentsconvexu}
\end{figure}
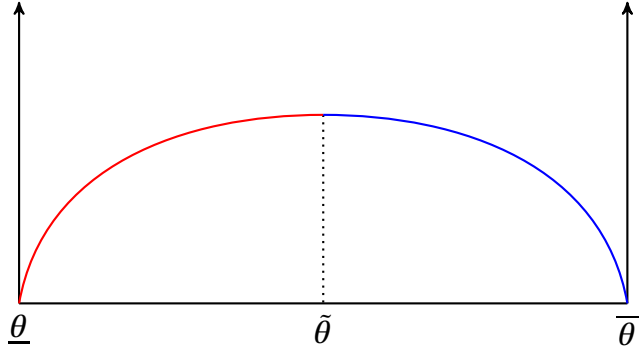

In the majority case, there will be a set of agents $I$ whose participation constraint is violated. Again, there will be at most two types for which the participation constraint is binding. Contrary to the concave case, the winning coalition is now a convex set. \Cref{fig:rentsconvexmex} illustrates this property.

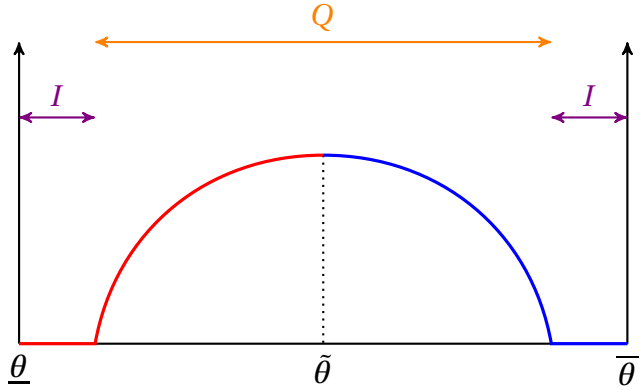
\begin{figure}[ht]
\centering
    \begin{tikzpicture}[-,>=stealth',auto,node distance=1.2cm,
  thick,main node/.style={transform shape, circle,draw,font=\rmfamily, minimum width=5pt}]

    \draw (0,2)node[below]{$\underline{\theta}$} -- (8,2);
    \draw[->] (0,2) -- (0,6);
    \draw[->] (8,2)node[below]{$\overline{\theta}$} -- (8,6);
    
   \draw (4,2) node[below]{$\Tilde{\theta}$};
   \draw[dotted] (4,4.5)--(4,2);

   \draw [red, very thick] (4,4.5)  to [out=180,in=80] (1,2)--(0,2);
   \draw [blue, very thick] (4,4.5) to [out=0, in=100] (7,2)--(8,2);


    \draw [<->] [orange] (1,6) -- (4,6)node[above]{$Q$} -- (7,6);
    \draw [<->] [violet] (0,5) -- (0.5,5)node[above]{$I$} -- (1,5);
    \draw [<->] [violet] (7,5) -- (7.5,5)node[above]{$I$} -- (8,5);

\end{tikzpicture}
\caption{$\Bar{v}_i$ convex with exclusion}
\label{fig:rentsconvexmex}
\end{figure}

\begin{proposition} \label{prop:convex}
    For any $g^\circ$, there exists two agents $\theta_p, \theta_q$ such that optimal public-good provision solves
    \begin{equation*}
    \left[ \theta_a + \sum_{\theta_i \in N\setminus \{a\}} \left(\theta_i - \frac{\gamma^*(\theta_i)-F_i(\theta_i)}{f_i(\theta_i)} \right) \right]\phi'(g(\theta)) = 1 \ ,
\end{equation*}
where
\begin{equation*}
    \gamma^*(\theta_i)=\begin{cases}
0 & \theta_i<\theta_p,\\
\alpha & \theta_p\le \theta_i<\theta_q,\\
1 & \theta_i\ge\theta_q \ .
\end{cases}
\end{equation*}
Moreover, the optimal coalition is given by
\begin{equation*}
 Q^*=\{a\} \cup \left\{\theta_i \in N \setminus \{a\} \ | \  \theta_{p} \leq \theta_i \leq \theta_q \right\} \ .    
\end{equation*}
\end{proposition}

\subsection{Dynamics}
Consider now the following dynamic setting. The agenda-setter $a \in N$ is the same in each period $\tau=1, \dots, T$ and s/he has full-commitment power. The agents' values for the public good do not evolve over time. Denote $\delta \in [0,1)$ the common discount factor of the agents. Moreover, let $g_{\tau}$ and $t_{\tau}$ be the public-good provision and transfer in period $\tau$, respectively. Agent $i$'s total utility is then given by
\begin{equation*}
    U_i^T(\theta_i)= \sum_{\tau=1}^T \delta^{\tau-1} (\theta_i \phi(g_{\tau}) - t_{\tau}) \ .
\end{equation*}
Analogously to \Cref{def:direct}, the notion of direct mechanism can be extended to this dynamic setting.
\begin{definition} \label{def:directdynamic}
A dynamic direct mechanism is a set of functions $g=(g_1, \dots, g_T)$ and $t_i=(t_{i,1}, \dots, t_{i,T})$, \quad  $\forall i \in N$
\begin{align*}
    g &: \Theta \longrightarrow \mathbb{R}_+^T,\\
    t_i &: \Theta \longrightarrow \mathbb{R}^T \ .
\end{align*}
\end{definition}
Moreover, as in \Cref{sec:sb}, the notions of incentive compatible and individually rational (static) mechanisms can be naturally extended to dynamic mechanisms.
\begin{definition}
    A dynamic direct mechanism is incentive compatible if, $\forall i \in N$,
    \begin{equation*}
       U_i^T(\theta) = \sum_{\tau=1}^T \delta^{\tau-1} (\theta_i \phi(g_{\tau}(\theta)) - t_{\tau}(\theta)) \geq \sum_{\tau=1}^T \delta^{\tau-1} (\theta_i \phi(g_{\tau}(\theta_i', \theta_{-i})) - t_{\tau}(\theta_i')) \ \qquad \forall \theta_i, \theta_i' \in [\underline{\theta}, \overline{\theta}] \ .
    \end{equation*} 
\end{definition}

\begin{definition} \label{def:dynamir}
    A dynamic direct mechanism is individually rational if, $\forall i \in N$,
    \begin{equation} \label{dynamicIR}
        U_i^T(\theta) \geq  \bar{v}_i(\theta_i, g^\circ) \beta \ \qquad \forall \theta_i \in [\underline{\theta}, \overline{\theta}] \ ,
    \end{equation}
    where $\beta= \frac{1-\delta^T}{1-\delta}$.
\end{definition}
In this dynamic setting, the agenda-setter proposes an intertemporal profile of public-good levels. If the agents vote in favour, then the public-good provision $g_{\tau}$ will be implemented in exchange for the transfer $t_\tau$ in each period $\tau$. If not, the outside option $ \Bar{v}_i(\theta_i, g^\circ)$ will be implemented in the current and future periods, without the possibility to renegotiate the public-good provision.

It is straightforward to show that the optimal dynamic public-good provision is simply a repetition of the solution to the static problems from Theorem \ref{thrm:u} and \ref{thrm:m}.

Intuitively, the optimal public-good provision and coalition formation in period $\tau$ will not affect the outside option in period $\tau+1$. Hence, all periods are identical to each other. But then, if there were a different solution to the dynamic problem, the discounted per-period value of this alternative solution would dominate the static solutions from Theorems \ref{thrm:u} and \ref{thrm:m}, contradicting their optimality. This discussion is summarized in the following proposition.\footnote{The argument is the same as in Chapter 11 in \cite{börgers2015}, to which I refer the reader for a proof.}

\begin{proposition} \label{prop:dynamic}
 The optimal dynamic public-good provision and coalition formation are a repetition of the solutions to the static problem \ref{gproblem}.
\end{proposition}

In the dynamic setting just considered, the solution does not exhibit any dynamic feature precisely because there is no dynamic linkage in public-good provision between periods.

A more interesting setting to consider is the following. Suppose that the public-good provision in period $\tau$ becomes the outside option public-good level in period $\tau+1$; i.e., $g_\tau=g^\circ_{\tau+1}$.
The individual rationality constraint (\ref{dynamicIR}) has to be rewritten as follows.
\begin{definition} \label{eq:dynamIR'}
    A dynamic direct mechanism is individually rational if, $\forall i \in N$,
    \begin{equation} \label{dynamicIR}
        U_i^T(\theta_i) \geq 
\overline{V}_i(\theta_i)
\ \qquad \forall \theta_i \in [\underline{\theta}, \overline{\theta}] \ ,
    \end{equation}
    where
\begin{equation*}
\overline{V}_i(\theta_i)= \sum_{\tau=1}^T\delta^{\tau-1}\bar{v}_i(\theta_i, g_{\tau-1}) \ .
\end{equation*}
\end{definition}
I refer to this dynamic setting as one with \textit{endogenous outside option}. This problem becomes significantly more to solve, as it is unclear \textit{ex ante} which participation constraints should bind, and which coalition might be optimal intertemporally.

Assuming that agents are fully myopic and do not discount future periods, the analysis is analogous to Problem \ref{g'problem}. However, the thresholds at which incentive constraints change have to be adapted to incorporate tomorrow's outside option, which is a function of today's public-good provision. A more thorough dynamic analysis is left for future research.

\section{Extensions} \label{sec:relaxss}

\subsection{Negatively-sloped Reservation Utility Profile} \label{sec:negslope}
It is reasonable to assume that the reservation utility profile might be negatively sloped; agents care relatively more about the public good precisely because they would be relative worse-off under the outside option.
Suppose $\Bar{v}_i'(\theta_i, g^\circ)<0$, for every $i \in N$. Then, condition (\ref{eq:envelope}) yields
\begin{equation*}
dU_i/d\theta_i=\phi(g(\theta))-\Bar{v}_i'(\theta_i,g^\circ)  \ > 0  
\end{equation*}
for all $\theta_i \in \Theta_i$. The dominating incentive is then always to understate the type. The equilibrium public-good provision and coalition formation are then the same as in \Cref{prop:unanimtd} and \Cref{prop:majoritd} when $\Bar{v}_i(\theta, g^\circ) \equiv 0$.

A more general relaxation of Assumption \ref{ass:singlecrossing} complicates the analysis, since general incentive compatibility constraints have to be taken into account. As a result, bunching might have to be introduced across disjoint subset of types.

\subsection{Informed Principal} \label{sec:informed}
The present paper is concerned with dominant-strategy implementation, thereby sidestepping the informed principal problem (\cite{myerson1983}, \cite{mylovanov2014}). Therefore, a very interesting extension entails optimal public-good provision and coalition formation in a Bayesian setting. The natural question relates to whether or not the agenda-setter can benefit from concealing his/her type. There are two competing forces.

On the one hand, along the same lines of \cite{Fleckinger2007}, concealment relaxes the participation constraints of the agents; hence, the agenda-setter might benefit from it, depending on the outside option public-good level.

On the other hand, the agenda-setter must commit to choosing a coalition for any type s/he might have. Therefore, it is not clear \textit{ex ante} whether s/he can benefit from concealment regardless of the type.

A more thorough analysis is left for future research.

\subsection{Stochastic Mechanisms} \label{sec:stochastic}
Mechanisms can be stochastic in two ways: for each agent, the inclusion in the winning coalition $Q$ (\textit{stochastic coalitions}); or the public-good provision $g$ (\textit{stochastic public-good provision}).
\subsubsection{Stochastic Coalitions}
Suppose now that the incentive compatibility constraint has to be satisfied only for the members of the coalition $Q \in \mathcal{Q}$.
We consider the following problem faced by the agenda-setter
\begin{equation*}  \tag{$\mathcal{S}$} \label{sproblem}
\begin{gathered}
\max_{g(\theta), \bold{t}(\theta), Q} \qquad \int \theta_a \phi(g(\theta)) -t(\theta_a) \ dF(\theta)\\
\text{subject to} \qquad \; g(\theta) \leq t(\theta_a) + \sum_{i \in N \setminus \{a\}}t(\theta_i),\\
\frac{d g(\theta)}{d \theta} \geq 0, \\
\text{and for some $Q \in \mathcal{Q}$}, \\
dU_j/d\theta_j=\phi(g(\theta))-\Bar{v}_j'(\theta_j,g^\circ), \\
U_j(\theta_j, g(\theta)) \geq \Bar{v}_j(\theta_j, g^\circ) \qquad \forall \; j \in Q.
\end{gathered}
\end{equation*}
\medskip\\
Problem \ref{sproblem} differs from Problem \ref{mproblem} only insofar as the individual rationality constraints of the agents excluded from the coalition are concerned. The resource constraint is relaxed, since the agents in $N \setminus Q$ can now be taxed to some upper bound $\Bar{\tau}$.
\begin{proposition}
The coalition $Q$ is randomly chosen. Within this set of agents, the solution from Problem \ref{gproblem} applies. 
\end{proposition}
\subsubsection{Stochastic Public-Good Provision}
Since the utility function is linear in types and quasilinear in the transfers, the agents are risk-neutral. Moreover, the agenda-setter 's objective function is concave in the public good. Therefore, it is straightforward to show that there is no gain in proposing a stochastic public-good provision.
\begin{proposition} \label{prop:stochasticg}
    Deterministic public-good provision is optimal.
\end{proposition}

\section{Conclusions} \label{sec:conclusions}
In this paper, I embedded a mechanism design problem about the provision of a public good into a political economy framework. In line with the previous contributions from the screening literature, the shape of the outside option crucially determines the optimal public-good provision and coalition formation. Several directions for future research remain to be pursued.

First, as mentioned in Section \ref{sec:informed}, dominant-strategy implementation circumvents the issue of the `informed principal' in the sense of \cite{myerson1983}. An analysis in Bayesian implementation would raise significant complications, as there is no guarantee that the revelation principle holds to begin with.\footnote{See \cite{mylovanov2014} and the references therein.} However, relaxing the requirement of dominant-strategy implementation constitutes a promising direction for future research.

Second, exploring public-good provision and coalition formation with limited commitment (even in a static setting), along the lines of \cite{doval2022}, would also constitute an interesting exercise, despite the well-known difficulties in analyzing dynamic mechanism design problems with endogenous outside option.

Finally, the paper considered a single dimension of private information. Multidimensional screening problems are well-known to pose severe difficulties (e.g., \cite{rochet1998}). Nonetheless, members of organizations are likely to be heterogeneous along different dimensions of private information. Making progress on extending the model to multiple dimensions would provide additional insights into the coalition-formation problem.

\pagebreak

\appendix

\section{VCG Mechanisms} \label{sec:fb}

We now consider the class of VCG mechanisms and see whether they can implement the `efficient' (utilitarian first-best) public-good level.\footnote{The discussion in this section follows \cite{börgers2015}.}
Is this allocation implementable? We first need two definitions.
\begin{definition}
The public-good level $g$ is \textit{efficient} if it solves
\begin{equation} \label{eq:fb}
    \sum_N \theta_i \phi'(g)=1 \ .
\end{equation}
\end{definition}
\begin{definition}
A direct mechanism $(g,t_i)_{i\in N}$ is called a `Vickrey-Clarke-Groves'(VCG) mechanism if $g$ is efficient and if for every $i \in N$ there is a function
\begin{equation*}
    \tau_i: \Theta_i \longrightarrow \mathbb{R}
\end{equation*}
such that
\begin{equation*}
    t_i(\theta)= -\sum_{j \neq i}u_j(\theta, g(\theta))+\tau_i(\theta_{-i}) \ .
\end{equation*}
\end{definition}
\begin{proposition}[\cite{krishna2001}]
Suppose that for every $i \in N$ the set $\Theta_i$ is a convex subset of a finite-dimensional Euclidean space. Moreover, assume that for every $i\in N$ the function $u_i(\theta_i, g(\theta))$ is a convex function of $\theta_i$. Suppose that  $(g, t_i)_{i \in N}$ is a dominant strategy incentive-compatible mechanism. Then, a direct mechanism with the same $(g(\theta), t'_i)_{i \in N}$ is dominant strategy incentive-compatible if and only if for every $i \in N$ and every $\theta_{-i}$ there is a number $\tau_i(\theta_{-i}) \in \mathbb{R}$ such that
\begin{equation*}
    t'_i(\theta)=t_i(\theta) + \tau_i(\theta_{-i}) \quad \forall \theta \in \Theta \ .
\end{equation*}
\end{proposition}
\begin{corollary}[\cite{börgers2015}]
Suppose that for every $i \in N$ the set $\Theta_i$ is a convex subset of a finite-dimensional Euclidean space. Moreover, assume that for every $i\in N$ the function $u_i(\theta, g(\theta))$ is a convex function of $\theta_i$. Suppose that  $(g, t_i)_{i \in N}$ is a dominant strategy incentive-compatible mechanism, and suppose that $g$ is efficient. Then, $(g, t_i)_{i \in N}$ is a VCG mechanism.
\end{corollary}
In our setting, the two convexity conditions are met. Hence, VCG mechanisms are the only mechanisms that make efficient public-good level decision rules dominant strategy incentive-compatible.\\
However, we have the following result.
\setcounter{theorem}{-1}
\begin{theorem} \label{thrm:VCG}
No VCG mechanism is implementable.
\end{theorem}
\begin{proof}
\ For simplicity, assume $\phi(g)=\ln(1+g).$
Consider the public good $g^*$ and the transfers $t^*_i(\theta, g^*)_{i \in N \setminus \{a\}}$ of a VCG mechanism. By standard arguments (\cite{green1979a}), transfers do not cover the cost of the public good and external funds are needed to balance the budget.
Now, pick the agent with the lowest type, and denote it $\theta_j$. Consider agent $\theta_j$'s variation in utility brought about by a slight decrease in $g^*$ offset by a decrease in $t^*_j$. Formally
\begin{align*}
    \theta_j\ln(1+g^*-\delta)-\theta_j\ln(1+g^*)&=\underbrace{t^*_j(\theta, g^*-\delta)-t^*_j(\theta, g^*)}_{\epsilon}\\
    \theta_j\ln(\sum_{i \in N} \theta_i - \delta)-\theta_j\ln(\sum_{i \in N} \theta_i) &= \epsilon \\
    \delta&=\sum_{i \in N} \theta_i (e^{\frac{\epsilon}{\theta_j}}-1)
\end{align*}
By construction, agent $j$'s utility is unchanged. All other agents $\theta_k$, with $k\in N \setminus\{a,j\}$, have a utility loss equal to $ \epsilon- \frac{\theta_k}{\theta_j}\epsilon $. Therefore, for the transfers $t^*_i(\theta, g^*-\delta)_{i \in N \setminus \{a\}}$ to be still incentive-compatible, the proposer needs to compensate them. Now, consider the variation in the proposer's utility.
\begin{align*}
    &\theta_a\ln(1+g^*-\delta)-\theta_a\ln(1+g^*) - \sum_{i \in N \setminus \{a\}} \left( \frac{\theta_i}{\theta_j}-1 \right) \epsilon -[(g^*-\delta)-g^*]\\
    &= \sum_{i \in N} \theta_i (e^{\frac{\epsilon}{\theta_j}}-1)  - \sum_{i \in N} \frac{\theta_i}{\theta_j}\epsilon \\
    &=  O(\epsilon) >0 \ ,
\end{align*}
where the last equality follows from the Taylor expansion of $e^x$ around $x=0$. This argument generalizes to arbitrary functions $\phi(g)$. Hence, $g^*$ is not implementable.
\end{proof}
\Cref{thrm:VCG} has a simple interpretation. Suppose that the efficient public-good level is known and it is set in the Constitution. Moreover, suppose that there is availability of external funds for the set of agents to implement the efficient public-good provision. Even in that case, it would not be possible to find a mechanism to implement that public-good level in dominant strategies.

\section{Proofs} \label{sec:appendix}

\subsection{Proof of \Cref{lemma:ic}}
\begin{proof}   
Fix $i \in N\setminus\{a\}$ and $\theta_{-i} \in \Theta_{-i}$, and write
\[
\hat\phi(\theta_i) \equiv \phi(g(\theta_i,\theta_{-i})),
\qquad
W_i(\theta_i) \equiv \theta_i \hat\phi(\theta_i) - t_i(\theta_i,\theta_{-i}),
\]
so that $U_i = W_i - \Bar{v}_i$.
 
\medskip
\noindent\emph{Necessity.} Take $\theta_i > \theta_i'$. Incentive
compatibility for type $\theta_i$ against report $\theta_i'$, and for type
$\theta_i'$ against report $\theta_i$, yields
\[
W_i(\theta_i) \ge W_i(\theta_i') + (\theta_i - \theta_i')\,\hat\phi(\theta_i'),
\qquad
W_i(\theta_i') \ge W_i(\theta_i) - (\theta_i - \theta_i')\,\hat\phi(\theta_i).
\]
Combining,
\[
\hat\phi(\theta_i)
\;\ge\;
\frac{W_i(\theta_i) - W_i(\theta_i')}{\theta_i - \theta_i'}
\;\ge\;
\hat\phi(\theta_i'),
\]
which delivers (i) --- and, since $\phi$ is strictly
increasing, monotonicity of $g(\cdot, \theta_{-i})$. Moreover,
\[
W_i(\theta_i)
= \sup_{\theta_i'} \bigl\{\theta_i \hat\phi(\theta_i')
- t_i(\theta_i',\theta_{-i})\bigr\}
\]
is a pointwise supremum of affine functions of $\theta_i$, hence convex,
hence absolutely continuous with $W_i'(\theta_i) = \hat\phi(\theta_i)$ at
every point of differentiability. Integrating, and subtracting
$\Bar{v}_i(\theta_i, g^\circ) - \Bar{v}_i(\ubar\theta, g^\circ)
= \int_{\ubar\theta}^{\theta_i} \Bar{v}_i'(x, g^\circ)\,dx$,
gives (ii).
 
\medskip
\noindent\emph{Sufficiency.} Suppose (i)--(ii) hold and take any
$\theta_i, \theta_i'$ with $\theta_i > \theta_i'$ (the reverse case is
symmetric). Adding $\Bar{v}_i(\theta_i,g^\circ) - \Bar{v}_i(\theta_i',g^\circ)$
to both sides of (ii) evaluated at $\theta_i$ and $\theta_i'$, the
reservation terms cancel and
\[
W_i(\theta_i) - W_i(\theta_i')
= \int_{\theta_i'}^{\theta_i} \hat\phi(x)\, dx
\;\ge\; (\theta_i - \theta_i')\,\hat\phi(\theta_i'),
\]
by monotonicity of $\hat\phi$. Rearranging,
$W_i(\theta_i) \ge \theta_i\hat\phi(\theta_i') - t_i(\theta_i',\theta_{-i})$:
type $\theta_i$ does not gain by reporting $\theta_i'$.
\end{proof}

\subsection{Proof of \Cref{lemma:ir}}
\begin{proof}
\ Fix $i$ and $\theta_{-i}$. By definition,
\[
U_i(\theta_i, \theta_{-i})
= U_i(\ubar\theta, \theta_{-i})
+ \int_{\ubar\theta}^{\theta_i}
\bigl[\phi(g(x, \theta_{-i})) - \phi(g^\circ)\bigr]\,dx \ .
\]

Then, $U_i(\cdot, \theta_{-i})$ is strictly convex, since it is the integral of a strictly increasing function.
Suppose $U_i$ is constant on a non-degenerate interval. For any $x$ in that interval, it must hold that
$\phi(g(x,\theta_{-i})) = \phi(g^\circ)$. Then, by strict
monotonicity of $\phi$, $g(x,\theta_{-i}) = g^\circ$, a contradiction. Hence, the set on which the participation constraints
bind is at most a singleton.
\end{proof}

\subsection{Proof of \Cref{prop:unanimtd}}
\begin{proof}
From Lemma \ref{lemma:ir}, it follows immediately that, if $g(\cdot,\theta_{-i}) > g^\circ$ everywhere, then $U_i (\cdot, \theta_{-i})$ is strictly increasing and it can be zero only at $\underline{\theta}$. Analogously, if  $g(\cdot,\theta_{-i}) < g^\circ$ everywhere, the participation constraint can only bind at $\overline{\theta}$.
Therefore, there are three cases to be considered.
\begin{itemize}
    \item[(a)] $g^{\circ} \leq g_U^L(\theta)$.\\
\medskip
The proposer's problem \ref{problem:lu} becomes
\begin{equation*}
\begin{gathered}
\max_{g(\theta), \bold{t}(\theta)} \qquad \int \theta_a \phi(g(\theta)) -t(\theta_a) \ dF(\theta) \\
\text{subject to} \; \; g(\theta) \leq t(\theta_a) + \sum_{i \in N\setminus \{a\}}t(\theta_i),\\  \frac{d g(\theta)}{d \theta} \geq 0, 
 \\ t(\theta_i)= \theta_i\phi(g(\theta_i, \theta_{-i})) - \int_{\underline{\theta}}^{\theta_i} \phi(g(x, \theta_{-i})) dx ,\\
t(\theta_i) \leq \theta_i \phi(g(\theta)) \qquad \forall \; i \in N\setminus\{a\}
\ .
\end{gathered}
\end{equation*}
\ Consider the proposer's expected utility:
\begin{align*}
   &\int_{\Theta} \left[ \theta_a \phi(g(\theta)) - \sum_{i \in N\setminus \{a\}}t(\theta_i) - g(\theta) \right] f(\theta) d\theta\\
   &= \int_{\Theta} \left[ \theta_a \phi(g(\theta)) + \sum_{i \in N\setminus \{a\}} \left[\theta_i\phi(g(\theta)) - \int_{\underline{\theta}}^{\theta_i} \phi(g(x, \theta_{-i})) dx \right] - g(\theta) \right] f(\theta) d\theta\\
    &= \int_{\Theta} \left[\sum_{i \in N} \theta_i\phi(g(\theta)) - \sum_{i \in N\setminus \{a\}} \left(\int_{\underline{\theta}}^{\theta_i} \phi(g(x, \theta_{-i})) dx \right) - g(\theta)\right] f(\theta) d\theta\\
    &= \int_{\Theta} \left[\sum_{i \in N} \theta_i\phi(g(\theta)) - \sum_{i \in N\setminus \{a\}} \left(\frac{1-F_i(\theta_i)}{f_i(\theta_i)} \right)\phi(g(\theta)) - g(\theta)\right] f(\theta) d\theta
\end{align*}
where the last equality follows from Fubini's theorem.
Pointwise maximization with respect to $g(\theta)$ yields
\begin{equation} \label{eq:unanimlow}
    \left[ \theta_a + \sum_{i \in N\setminus \{a\}} \left(\theta_i - \frac{1-F_i(\theta_i)}{f_i(\theta_i)} \right)\right]\phi'(g(\theta)) = 1 \ .
\end{equation}
\Cref{assumpt:regular} ensures the monotonicity constraint is satisfied.
The same goes through for any $g^\circ \leq g^L_U(\theta)$.\\
\item[(b)] $g^{\circ} > g^H_U(\theta)$.\\
The proposer's problem \Cref{problem:lu} becomes
\begin{equation*}
\begin{gathered}
\max_{g(\theta), \bold{t}(\theta)} \qquad \int \theta_a \phi(g(\theta)) -t(\theta_a) \ dF(\theta) \\
\text{subject to} \; \; g(\theta) \leq t(\theta_a) + \sum_{i \in N\setminus \{a\}}t(\theta_i),\\  \frac{d g(\theta)}{d \theta} \geq 0, 
 \\ t(\theta_i)= \theta_i\phi(g(\theta_i, \theta_{-i})) - \int_{\theta_i}^{\overline{\theta}} \Bar{v}_i(x, g^\circ) - \phi(g(x, \theta_{-i})) \ dx ,\\
\theta_i \phi(g(\theta)) -t(\theta_i) \geq \theta_i \phi(g^\circ)-\frac{g^\circ}{n} \qquad \forall \; i \in N \setminus \{a\}\ .
\end{gathered}
\end{equation*}
\ By an analogous reasoning as above, we have
\begin{align*}
   &\int_{\Theta} \left[ \theta_a \phi(g(\theta)) - \sum_{i \in N\setminus \{a\}}t(\theta_i) - g(\theta) \right] f(\theta) d\theta\\
   &= \int_{\Theta} \left[ \theta_a \phi(g(\theta)) + \sum_{i \in N\setminus \{a\}} \left[\theta_i\phi(g(\theta)) - \int_{\theta_i}^{\overline{\theta}} \Bar{v}_i(x, g^\circ) - \phi(g(x, \theta_{-i})) \ dx \right] - g(\theta) \right] f(\theta) d\theta\\
    &= \int_{\Theta} \left[\sum_{i \in N} \theta_i\phi(g(\theta)) - \sum_{i \in N\setminus \{a\}} \left(\int_{\theta_i}^{\overline{\theta}} \Bar{v}_i(x, g^\circ) - \phi(g(x, \theta_{-i})) \ dx \right) - g(\theta)\right] f(\theta) d\theta\\
    &= \int_{\Theta} \left[\sum_{i \in N} \theta_i\phi(g(\theta)) + \sum_{i \in N\setminus \{a\}} \left(\frac{F_i(\theta_i)}{f_i(\theta_i)} \right)\phi(g(\theta)) - \sum_{i \in N\setminus \{a\}} \left(\int_{\theta_i}^{\overline{\theta}} \Bar{v}_i(x, g^\circ)\right) - g(\theta)\right] f(\theta) d\theta \ .
\end{align*}
Pointwise maximization with respect to $g(\theta)$ yields
\begin{equation} \label{eq:unanimti}
    \left[ \theta_a + \sum_{i \in N\setminus \{a\}} \left(\theta_i + \frac{F_i(\theta_i)}{f_i(\theta_i)} \right)\right]\phi'(g(\theta)) = 1 \ .
\end{equation}
Again, by \Cref{assumpt:regular}, the monotonicity constraint is satisfied. Moreover, the same goes through for any $g^\circ \geq g^H_U(\theta)$.\\

\item[(c)]  $g^L< g^\circ < g^H$.\\
By Lemma \ref{lemma:ir}, the participation constraint can bind on a single type. If it is an extreme type, the optimal public-good provisions are as in the previous cases. But then, incentive compatibility is violated \textit{ex post}. Suppose now that the participation constraint is binding on an interior type $\Tilde{\theta}$. Since $\phi$ is strictly increasing, information rents have to be either increasing (decreasing) for all types. But then, the participation constraints of types below (above) $\Tilde{\theta}$ are violated. Therefore, the agenda-setter can do no better than proposing the outside option public-good level.

\end{itemize}
\end{proof}

\subsection{Proof of \Cref{prop:majoritd}}
\begin{proof}
\ We begin with the following two lemmata.
\begin{lemma} \label{lemma:coalition}
The optimal choice of coalition is
\[
Q^{*}=
\begin{cases}
\overline{Q} & \text{if } g^\circ<g^L_M(\theta),\\
\text{any } Q\in\mathcal{Q} &
\text{if } g^L_M(\theta)\le g^\circ\le g^H_M(\theta),\\
\underline{Q} & \text{if } g^\circ>g^H_M(\theta),
\end{cases}
\]
if $g^L_M(\theta)<g^H_M(\theta)$. Otherwise there exists
$\hat g^\circ$ such that $Q^{*}=\overline{Q}$ for $g^\circ<\hat g^\circ$ and
$Q^{*}=\underline{Q}$ for $g^\circ>\hat g^\circ$. Moreover, $|I(\theta)|=n-q$ for almost every $\theta$ at which
$g^{*}(\theta)\neq g^\circ$.
\end{lemma}
\begin{proof}
\ By Lemma \ref{lemma:ir}, the participation constraints binds on a single type. For a given $g^\circ$, denote $\Pi(g(\theta), \Tilde{\theta_i}; g^\circ)$ the generic payoff of the agenda-setter when the participation constraint binds on $\Tilde{\theta_i}$. Now, fix the public-good level $g(\theta)$ and raise the binding participation constraint to $\Tilde{\theta}_{i+1}$. The variation in payoff is equal to
\begin{equation*}
    \Pi(g(\theta), \Tilde{\theta_i}; g^\circ) - \Pi(g(\theta), \Tilde{\theta}_{i+1}; g^\circ)= \int_{\Tilde{\theta_i}}^{\Tilde{\theta}_{i+1}} [\phi(g(x, \theta_{-i}))-\phi(g^\circ)] \ dx \geq 0 \ .
\end{equation*}

Now let $V(\Tilde{\theta}_i)=\max_{g(\theta)} \Pi(g(\theta), \Tilde{\theta_i}; g^\circ)$, and denote $g^*(\theta)$ its maximizer. Then
\begin{equation*}
    V(\Tilde{\theta}_{i+1}) \geq \Pi(g^*(\theta), \Tilde{\theta_i}; g^\circ)= V(\Tilde{\theta}_i) + \int_{\Tilde{\theta_i}}^{\Tilde{\theta}_{i+1}} [\phi(g(x, \theta_{-i}))-\phi(g^\circ)] \geq V(\Tilde{\theta}_i) \ .
\end{equation*}
\end{proof}
With the help of \Cref{lemma:coalition}, the analogous cases as in \Cref{prop:unanimtd} now follow.
\paragraph{Case I} $g^\circ < g^L_M(\theta)$.\\
The agenda-setter's problem \ref{problem:lm} becomes
\begin{equation*}
\begin{gathered}
\max_{g(\theta), \bold{t}(\theta)} \qquad \int \theta_a \phi(g(\theta)) -t(\theta_a) \ dF(\theta) \\
\text{subject to} \; \; g(\theta) \leq t(\theta_a) + \sum_{i \in N\setminus \{a\}}t(\theta_i),\\  \frac{d g(\theta)}{d \theta} \geq 0, 
 \\ t(\theta_i)= \theta_i\phi(g(\theta_i, \theta_{-i})) - \int_{\Tilde{\theta}}^{\theta} \phi(g(x, \theta_{-i})) dx \qquad \forall \; i \in N\setminus\{a\},\\
t(\theta_i) \leq \theta_i \phi(g(\theta)) \qquad \forall \; i \in Q
\ .
\end{gathered}
\end{equation*}
By \Cref{lemma:coalition}, it is then optimal to make the participation constraint of $\Tilde{\theta}$ binding.
\begin{claim}
    For all $\theta_i \leq \Tilde{\theta}$, the participation constraint is violated.
\end{claim}
\begin{proof}
\ We have that
\begin{align*}
    t(\theta_i)&= \theta_i\phi(g(\theta)) - \int_{\Tilde{\theta}}^{\theta_i} \phi(g(x, \theta_{-i})) dx \\
    &=  \theta_i\phi(g(\theta)) + \int_{\theta_i}^{\Tilde{\theta}} \phi(g(x, \theta_{-i})) dx \\
    &\geq \theta_i\phi(g(\theta)) \ ,
\end{align*}
where the last inequality follows from the monotonicity of the allocation rule.
\end{proof}
Consider now the expected transfer of types $\theta_i \geq \Tilde{\theta}$. It is equal to
\begin{equation*}
    \int_{\Tilde{\theta}}^{\overline{\theta}} t(\theta_i) \ f_i(\theta_i) d(\theta_i)= \int_{\Tilde{\theta}}^{\overline{\theta}} \theta_i\phi(g(\theta_i, \theta_{-i})) - \int_{\Tilde{\theta}}^{\overline{\theta}} \int_{\Tilde{\theta}}^{\theta} \phi(g(x, \theta_{-i})) dx \ f_i(\theta_i) d(\theta_i) \ .
\end{equation*}
The second term is equal to
\begin{align*}
   &\int_{\underline{\theta}}^{\overline{\theta}}  \int_{\Tilde{\theta}}^{\theta} \phi(g(x, \theta_{-i})) dx \ f_i(\theta_i) d(\theta_i) \\
   &= \int_{\Tilde{\theta}}^{\theta} \int_{\theta}^{\overline{\theta}} \phi(g(x, \theta_{-i})) \ f_i(\theta_i) d(\theta_i) \ dx \\
   &= \int_{\Tilde{\theta}}^{\theta} \phi(g(\theta)) [ 1-F_i(\theta_i)] \ d \theta_i \ .
\end{align*}

The expected transfers of the excluded types $\theta_i < \Tilde{\theta}$ are equal to
\begin{align*}
    \int_{\underline{\theta}}^{\Tilde{\theta}} t(\theta_i) \ f_i(\theta_i) d(\theta_i) &= \int_{\underline{\theta}}^{\Tilde{\theta}} \theta_i\phi(g(\theta_i, \theta_{-i})) - \int_{\underline{\theta}}^{\Tilde{\theta}} \int_{\Tilde{\theta}}^{\theta} \phi(g(x, \theta_{-i})) dx \ f_i(\theta_i) d(\theta_i) \\
    &= \int_{\underline{\theta}}^{\Tilde{\theta}} \theta_i\phi(g(\theta_i, \theta_{-i})) + \int_{\underline{\theta}}^{\Tilde{\theta}} \int_{\theta}^{\Tilde{\theta}} \phi(g(x, \theta_{-i})) dx \ f_i(\theta_i) d(\theta_i) \ .
\end{align*}

The second term of the last equation is equal to
\begin{align*}
    &\int_{\underline{\theta}}^{\Tilde{\theta}} \int_{\theta}^{\Tilde{\theta}} \phi(g(x, \theta_{-i})) dx \ f_i(\theta_i) d(\theta_i) \\
    &= \int_{\underline{\theta}}^{\Tilde{\theta}} \int_{\underline{\theta}}^{\theta} \phi(g(x, \theta_{-i})) \ f_i(\theta_i) d(\theta_i) \ dx \\
    &= \int_{\underline{\theta}}^{\Tilde{\theta}} \phi(g(\theta)) F(\theta) \ d \theta_i \ .
\end{align*}

The two simplifications of the information rents stem from Fubini's theorem; namely exchanging the order of integration in the red and blue areas, respectively.\\
\begin{figure}[ht]
\centering
    \begin{tikzpicture}[-,>=stealth',auto,node distance=1.2cm,
  thick,main node/.style={transform shape, circle,draw,font=\rmfamily, minimum width=5pt}, scale=0.8]

    \draw[->] (-5,0) -- (5,0) node[below]{$\theta_i$};
    \draw[->] (-4.5,-0.5) -- (-4.5,7) node[left]{$x$};
    
    \draw (-3.5,1) -- (3.5,5.5);
    \draw (3.5,1) -- (3.5,5.5);
    \draw (-3.5,1) -- (3.5,1);

    \draw (-3.5,0) node[below]{$\underline{\theta}$};

    \draw (-4.5,1) node[left]{$\underline{\theta}$};

    \draw (3.5,0) node[below]{$\overline{\theta}$};

    \draw (-4.5,5.5) node[left]{$\overline{\theta}$};

    \draw (-1.5,0) node[below]{$\Tilde{\theta}$};

    \draw (-4.5,2.25) node[left]{$\Tilde{\theta}$};

    \draw[very thick, red] (-1.5,2.25) -- (3.5, 5.5) -- (3.5,2.25)  -- (-1.5,2.25);

    \draw[very thick, blue] (-3.5,1) -- (-3.5, 2.25) -- (-1.5,2.25)  -- (-3.5,1);

    
    \draw[dashed, color=gray] (3.5, 1.6675) -- (-2.3, 1.6675);
    \draw[dashed, color=gray] (3.5, 2.25) -- (-1.5, 2.25);
    \draw[dashed, color=gray] (3.5, 2.85) -- (-0.7, 2.85);
    \draw[dashed, color=gray] (3.5, 3.5) -- (0.5, 3.5);
    \draw[dashed, color=gray] (3.5, 4.2) -- (1.5, 4.2);
    \draw[dashed, color=gray] (3.5, 4.9) -- (2.5, 4.9);

    \draw[dashed, color=gray] (-2.3, 1) -- (-2.3, 1.6675);
    \draw[dashed, color=gray] (-1.5, 1) -- (-1.5, 2.25);
    \draw[dashed, color=gray] (-0.7, 1) -- (-0.7, 2.85);
    \draw[dashed, color=gray] (0.5, 1) -- (0.5, 3.5);
    \draw[dashed, color=gray] (1.5, 1) -- (1.5, 4.2);
    \draw[dashed, color=gray] (2.5, 1) -- (2.5, 4.9);

\end{tikzpicture}
\caption{Change in the order of integration.}
\end{figure}
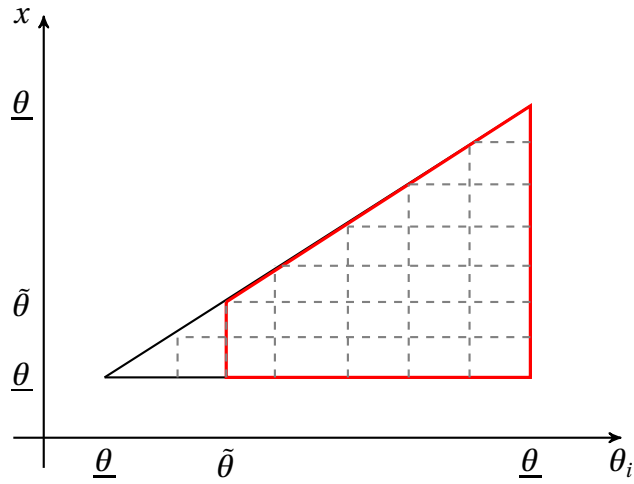

Putting everything together, and going through the same derivations as for \Cref{prop:unanimtd}, yields the following result.
\begin{lemma}
   Optimal public-good provision solves
\begin{equation}
   \left[ \theta_a + \sum_{i \in I} \left(\theta_i + \frac{F_i(\theta_i)}{f_i(\theta_i)} \right) + \sum_{i \in \overline{Q}\setminus \{a\}} \left(\theta_i - \frac{1-F_i(\theta_i)}{f_i(\theta_i)} \right)\right]\phi'(g(\theta)) = 1 \ .
\end{equation} 
\end{lemma}
\paragraph{Case II} $g^\circ > g_M^H(\theta)$.\\
Along the same lines as in the previous case, the following result holds.
\begin{lemma}
   Optimal public-good provision solves
\begin{equation}
    \left[ \theta_a + \sum_{i \in I} \left(\theta_i - \frac{1-F_i(\theta_i)}{f_i(\theta_i)} \right) + \sum_{i \in \underline{Q}\setminus \{a\}} \left(\theta_i + \frac{F_i(\theta_i)}{f_i(\theta_i)} \right)\right]\phi'(g(\theta)) = 1  \ .
\end{equation} 
The same goes through for any $g^\circ \geq g^H_M(\theta)$.
\end{lemma}
\paragraph{Case III} $g^L_M(\theta) < g^\circ < g^H_M(\theta)$.\\
The argument is the same as in Case III in \Cref{prop:unanimtd}.
\end{proof}

\subsection{Proof of \Cref{thrm:u}}
\subsubsection{Proof of \Cref{lemma:objfct}}
\begin{proof}
\ By the same argument as in \cite{jullien2000}, the measure $\dd\gamma_i^*$ is the shadow value of agent $i$'s participation
constraints. Since the resource constraint is binding, the Lagrangian of Problem  \ref{gproblem} can be rewritten as 
\[
\mathcal{L}
= \int_{\Theta}\theta_a\phi(g(\theta)  + \sum_{i\neq a}
\Bigl[\mathbb{E}_{\theta_{-i}}W_i(\theta_i,\theta_{-i})
- \vbar_i(\theta_i,g^\circ)\Bigr]\dd\gamma_i^*(\theta_i) \ f(\theta) d(\theta) \ .
\]

As in the proofs of Propositions \ref{prop:unanimtd} and \ref{prop:majoritd}, if the participation constraint is binding on some type $\Tilde{\theta}$, by applying Fubini's theorem we derive

\begin{align*}
\int W_i \dd\gamma_i^*(\theta_i)
&= W_i(\Tilde{\theta,}\theta_{-i})
+ \int_{\underline{\theta}}^{\Tilde{\theta}}
\phi(g(x,\theta_{-i}))\,\bigl[\gamma_i^*(x^-)\bigr]\dd x + \int_{\Tilde{\theta}}^{\overline{\theta}}
\phi(g(x,\theta_{-i}))\,\bigl[1-\gamma_i^*(x^-)\bigr]\dd x ,\\
\int W_i\, f_i(\theta_i)\dd\theta_i
&=W_i(\Tilde{\theta,}\theta_{-i})
+ \int_{\underline{\theta}}^{\Tilde{\theta}}
\phi(g(x,\theta_{-i}))\,\bigl[F_i(x)\bigr]\dd x + \int_{\Tilde{\theta}}^{\overline{\theta}}
\phi(g(x,\theta_{-i}))\,\bigl[1-F_i(x)\bigr]\dd x .
\end{align*}
Subtracting, summing over
$i\in N\setminus\{a\}$, and substituting into $\mathcal{L}$ yields the desired expression.
\end{proof}
\begin{proof}
By Lemma \ref{lemma:objfct}, Problem \ref{gproblem} can be rewritten as
\begin{equation} \tag{$\mathcal{G'}$} \label{g'problem}
\begin{gathered}
\max_{W(\Tilde{\theta}), g(\theta)} \qquad \int \sigma(\theta, \gamma(\theta), g(\theta); g^\circ) \ dF(\theta)\\
\text{subject to} \qquad 
g(\theta_i, \theta_{-i}) \ \text{increasing in} \ \theta_i , \quad g(\theta) \geq 0 \ .
\end{gathered}
\end{equation}
The participation constraints bind according to condition $\mathcal{A}$; hence, the problem reduces to choosing $g$. Since the hazard rates are bounded, the objective function is bounded as well: By Helly Selection Theorem, for any bounded sequence $\{g_n\}$, there exists a monotone (and hence implementable) bounded subsequences $g_{nk}$ converging pointwise to $g^*(\theta)$.  By Lebesgue's Dominated Convergence Theorem, $g^*$ yields $\Pi(W^*, g^*) \geq \limsup{\Pi(g_n)}= \sup{\Pi}$. Therefore, it is optimal. Finally, uniqueness follows from the strict concavity of the objective function.
\end{proof}
\subsection{Proof of \Cref{thrm:m}}
Proving that the agenda-setter always benefits from excluding an extra agent is analogous to the proof of \Cref{lemma:coalition}, once the general functional form $\Bar{v}_i$ is taken into account.

\subsection{Proof of Section \ref{sec:implications}}

\begin{proof}[Proof of \Cref{lemma:propofG}]
By Assumption \ref{ass:singlecrossing}, $\partial G/\partial g^\circ
=-\frac{\partial^2 \Bar{v}_i}{\partial \theta_i \partial g^\circ}<0$: $G(\theta_i;\cdot)$ is
continuous and strictly decreasing, so on its range the solution
exists, is unique, and satisfies $\partial\bar g/\partial x
=1/\partial G/\partial g^\circ<0$. If $\Bar{v}_i$ is concave,
$\vbar_i'$ is non-increasing in $\theta_i$, so
$\partial G/\partial \theta_i \geq 0$  and
$\partial\bar g/\partial\theta_i=-\partial G/\partial \theta_i/\partial G/\partial g^\circ \geq 0$.
The convex case is analogous.
\end{proof}

\subsection{Concavity} \label{sec:concavity}
I begin with the proof of \Cref{lemma:partition}.
\begin{proof}
\ Consider $\theta_k$, $\theta_l \in L(g^\circ)$ and $\theta_m$. We have that
\begin{align*}
    dU_m/\theta_m = \phi(g(\theta)) - \Bar{v}'_m(\theta_m, g^\circ) &\geq \phi(g(\theta)) - \Bar{v}'_l(\theta_l, g^\circ)\\
    &=dU_l/d\theta_l\\
    &=0 \\
                  &\geq \phi(g(\theta)) - \Bar{v}'_k(\theta_k, g^\circ) \\
                  &= dU_k/d\theta_k \ .
\end{align*}
Hence, $\theta_k \leq \theta_l \leq \theta_m$.

To characterize the set of excluded agents as a function of $g^\circ$, define $J(\gamma^*)$ to be the utility differential between the highest and the lowest type
\begin{equation*}
    J(\gamma^*) = U(\overline{\theta}) - U(\underline{\theta}) 
                = \int_{\underline{\theta}}^{\overline{\theta}} \phi(g(\theta, \gamma^*)) - \Bar{v}'_i(\theta_i, g^\circ) \ d\theta \ .
\end{equation*}
Clearly, $J(\gamma^*)$ is decreasing. To determine the optimal choice of 
$\gamma^*$, we have the following lemma.
\begin{lemma}[\cite{maggi1995}] \label{lemma:R}
The optimal $\gamma^*$ satisfies
\begin{equation*}
    \gamma^* \equiv \begin{cases}
        1 &\mbox{if $J(1) \geq 0$}\\
        J^{-1}(0) &\mbox{if $J(1)<0<J(0)$}\\
        0 &\mbox{if $J(0) \leq 0$} \ .
    \end{cases}
\end{equation*}
\end{lemma}
\begin{proof}
\ When $J(1) \geq 0$ or $J(0) \leq 0$, the same reasoning as in \Cref{prop:unanimtd} applies. Consider now the case $J(1)<0<J(0)$. Choosing $\alpha=J^{-1}(0)$ ensures that $J(\alpha)=0$ and incentive compatibility is preserved.
Moreover, denote
\begin{equation}
    J(\theta; g^\circ) \equiv \int_{\underline{\theta}}^{\theta} \phi(g(x; \gamma^*)) - \Bar{v}'_i(x, g^\circ) \ dx \ .
\end{equation}
We have the following lemma.
\begin{lemma} \label{lemma:jconvex}
    $J(\theta;g^\circ)$ is convex in $\theta$.
\end{lemma}
\begin{proof}
  \  $J'(\theta; g^\circ)=0$ when $G(\theta_i; g^\circ)=0$, which occurs for a single $\theta_i$. By convexity of $\Bar{v}_i$, $J''(\theta; g^\circ)$ is increasing in $\theta$.
\end{proof}
Then, $U(\theta)=J(\theta;g^\circ)+U(\underline{\theta})$ and, by quasiconvexity, $U(\theta) \leq \min\left\{U(\underline{\theta}), U(\overline{\theta})\right\}$ for all $\theta_i \in \Theta_i$. Hence, if the participation constraint is satisfied for the highest and lowest types, it is violated for all types in between. By Theorem \ref{thrm:m}, then it is optimal to choose $\theta_p$ and $\theta_q$ so as to include no more than $q-1$ agents.
\end{proof}
\end{proof}
\subsection{Proof of \Cref{prop:concave}}
I now define a partition of the support of $g^\circ$ according to the agents' incentive to understate or overstate their types. 
In order to keep notation compact and consistent with \Cref{sec:linear}, define the following thresholds recursively as 
\begin{align*}
  \bar{g}^L &\equiv \left[ \theta_a + \sum_{i \in L,M} \left(\theta_i - \frac{1-F_i(\theta_i)}{f_i(\theta_i)} \right)\right]\phi'(g(\theta)) = 1 , \\
    \bar{g}^{L-1} &\equiv \left[ \theta_a + \left(\theta_1 + \frac{F_1(\theta_1)}{f_1(\theta_1)} \right)+ \sum_{i \in L,M} \left(\theta_i - \frac{1-F_i(\theta_i)}{f_i(\theta_i)} \right)\right]\phi'(g(\theta)) = 1 , \\
    &\vdots \\
    \bar{g}^{H-1} &\equiv \left[ \theta_a +  \sum_{i \in K} \left(\theta_i + \frac{F_i(\theta_i)}{f_i(\theta_i)} \right) + \left(\theta_r - \frac{1-F_r(\theta_r)}{f_r(\theta_r)} \right) \right]\phi'(g(\theta)) = 1 , \\
     \bar{g}^H &\equiv \left[ \theta_a + \sum_{i \in L,K} \left(\theta_i + \frac{F_i(\theta_i)}{f_i(\theta_i)} \right)\right]\phi'(g(\theta)) = 1  \ .
\end{align*}
Then, \Cref{lemma:partition} and \Cref{lemma:propofG} imply the following proposition.
\begin{proposition} \label{prop:kl}
Suppose $\Bar{v}_i''(\theta_i, g^\circ) \leq 0$. The indices $x = x(g^\circ)$ and $y = y(g^\circ)$ are given by:
\begin{align*}
    x(g^\circ) &= \left\{ \begin{array}{lcl} 
    0 & \mbox{if} & 0 \leq g^\circ < \bar{g}^L \\ 
    1 & \mbox{if} & \bar{g}^L \leq g^\circ < \bar{g}^{L-1} \\
    \;\vdots & \mbox{if} & \qquad \vdots \qquad \vdots \qquad \\
    r-(q-2) & \mbox{if} & \bar{g}^{H-1} \leq g^\circ < \bar{g}^H \\
    r-(q-1) & \mbox{if} & \bar{g}^H \leq g^\circ \ ,
    \end{array}\right.\\
    &\text{ }\\
    y(g^\circ) &= \left\{ \begin{array}{lcl} 
    r-q & \mbox{if} & 0 \leq g^\circ < \bar{g}^L \\  
    r-(q-1) & \mbox{if} & \bar{g}^L \leq g^\circ < \bar{g}^{L-1} \\
    \;\vdots & \mbox{if} & \qquad \vdots \qquad \vdots \qquad \\
    r-1 & \mbox{if} & \bar{g}^{H-1} \leq g^\circ < \bar{g}^H \\
    r & \mbox{if} & \bar{g}^H \leq g^\circ \ .
    \end{array}\right.
\end{align*}
\end{proposition}
\begin{proof}
\begin{itemize}
    \item $x=0$, $y=r-q$. Same as in \Cref{prop:unanimtd}.
    \item $x>0$, $y<r$. By \Cref{lemma:partition}, the set of types is ordered. As shown above, rents are convex and thus reach a minimum at a single interior type. Denote this type by $\Tilde{\theta}$. Formally, the proposer's problem becomes
\begin{equation*}
\begin{gathered}
\max_{g(\theta)} \qquad \int \sigma(\theta, \gamma^*(\theta), g(\theta); g^\circ) \ dF(\theta)\\
\text{subject to} \qquad 
g(\theta_i, \theta_{-i}) \ \text{increasing in} \ \theta_i , \quad g(\theta) \geq 0 \ ,
\end{gathered}
\end{equation*}
where $\gamma^*(\theta)$ is the one derived in Lemma \ref{lemma:R}.
Then, going through the same derivations as in \Cref{prop:unanimtd} yields the desired condition from Proposition \ref{prop:concave}.
    \item $x=r-(q-1)$, $y=r$. Same as in \Cref{prop:unanimtd}.
    \end{itemize}
\end{proof}

\subsection{Convexity} \label{sec:convexity}
I begin with the proof of \Cref{lemma:partition2}.\\
\begin{proof}
\ Consider $\theta_k$, $\theta_l \in L(g^\circ)$ and $\theta_m$. We have that
\begin{align*}
    dU_m/\theta_m = \phi(g(\theta)) - \Bar{v}'_m(\theta_m, g^\circ) &\geq \phi(g(\theta)) - \Bar{v}'_l(\theta_l, g^\circ) \\
                  &=dU_l/d\theta_l\\
                  &=0 \\
                  &\geq \phi(g(\theta)) - \Bar{v}'_k(\theta_k, g^\circ) \\
                  &= dU_k/d\theta_k \ .
\end{align*}
Hence, $\theta_m \leq \theta_l \leq \theta_k$.
However, in contrast to \Cref{lemma:partition}, there cannot be a single interior type in $L(g^\circ)$.

Intuitively, as $g^\circ$ increases, the agents with the highest types transition into $K(g^\circ)$. By Condition (\ref{eq:envelope}), information rents have to be decreasing in $K(g^\circ)$ and increasing in $M(g^\circ)$. Since $\gamma^*(\theta_i)$ is nondecreasing, it must be constant for all types; i.e., $\gamma^*(\theta_i)=\gamma^*$, for all $i \in N \setminus \{a\}$.

Finally, by a similar argument as in Lemma \ref{lemma:jconvex}, $J(\theta; g^\circ)$ is now concave. Then, $U(\theta)=J(\theta;g^\circ)+U(\theta_p)$ and, by quasiconvexity, $U(\theta) \leq \max \left\{U(\theta_p), U(\theta_q)\right\}$ for all $\theta \in [\theta_p, \theta_q]$. Hence, if the participation constraint is satisfied for $\theta_p$ and $\theta_q$, it is violated for all types in between.

This concludes the proof of \Cref{lemma:partition2}.\\
\medskip\\
As in \Cref{sec:concavity}, it is useful to define the following thresholds of the support of $g^\circ$ recursively as follows:
\begin{align*}
  \bar{g}^L &\equiv \left[ \theta_a + \sum_{i \in L,M} \left(\theta_i - \frac{1-F_i(\theta_i)}{f_i(\theta_i)} \right)\right]\phi'(g(\theta)) = 1  \\
    \bar{g}^{L-1} &\equiv \left[ \theta_a + \left(\theta_r - \frac{\gamma^*_{L-1} -F_r(\theta_r)}{f_r(\theta_r)} \right)+ \sum_{i \in L,M} \left(\theta_i - \frac{\gamma^*_{L-1}-F_i(\theta_i)}{f_i(\theta_i)} \right)\right]\phi'(g(\theta)) = 1  \\
    \bar{g}^{H-1} &\equiv \left[ \theta_a +  \sum_{i \in K} \left(\theta_i - \frac{\gamma^*_{H-1} - F_i(\theta_i)}{f_i(\theta_i)} \right) + \left(\theta_1 - \frac{\gamma^*_{H-1}-F_1(\theta_1)}{f_r(\theta_1)} \right) \right]\phi'(g(\theta)) = 1 \\
     \bar{g}^H &\equiv \left[ \theta_a + \sum_{i \in L,K} \left(\theta_i + \frac{F_i(\theta_i)}{f_i(\theta_i)} \right)\right]\phi'(g(\theta)) = 1  \ .
\end{align*}
Notice that $\gamma^*$ changes with the outside option public-good level, as shown in \Cref{lemma:R}.\\
\Cref{lemma:partition2} and \Cref{lemma:propofG} imply the optimality of the same indices as in Proposition \ref{prop:kl}. The proof is analogous.
\end{proof}

\subsection{Proof of \Cref{prop:stochasticg}} \label{sec:stochasticg}
As in \cite{jullien2000}, define $g(\theta)$ as the public-good level such that
\begin{equation*}
    \phi(g(\theta)) = \mathbb{E}_{\mu} \left[ \phi(g(\theta)) \ | \ \theta \right] \ ,
\end{equation*}
where $\mathbb{E}_{\mu} \left[ \cdot \ | \ \theta \right]$ is the conditional expectation operator for the distribution $\mu(\theta, dg)$.

\begin{proof}
    \ Recall condition (\ref{eq:envelope}) and consider two types, $\theta_1$ and $\theta_2$. We have that
    \begin{align*}
        \int_{\theta_1}^{\theta_2} \phi(g(\theta)) - \Bar{v}'_i(\theta_i, g^\circ) \ d\theta_i &\geq \int_{\theta_1}^{\theta_2} \int_g \phi(g(\theta)) \ \mu(\theta_1, dg)) \ d\theta_i - \int_{\theta_1}^{\theta_2} \Bar{v}'_i(\theta_i, g^\circ) \ d\theta_i \\
        &= \int_{\theta_1}^{\theta_2} \phi(g(\theta_1, \theta_{-i})) \ d\theta_i - \Bar{v}'_i(\theta_i, g^\circ) \ d\theta_i \ ,
    \end{align*}
    where the first inequality follows from Jensen's inequality and the last equality by the linearity in types of the utility function. Therefore, $g(\theta)$ is incentive compatible.\\
    Moreover, by concavity of the objective function, we have that
    \begin{equation*}
       \sigma(\theta, \gamma(\theta), g) \geq \sigma \left(\theta, \gamma(\theta), \mathbb{E}_{\mu} [g] \right) .
    \end{equation*}
\end{proof}

\pagebreak

\bibliographystyle{apalike}
\bibliography{mdpolitician}

\begin{thebibliography}{}

\bibitem[Alesina and Perotti, 1994]{alesina94}
Alesina, A. and Perotti, R. (1994).
\newblock The political economy of growth: A critical survey of the recent literature.
\newblock {\em The World Bank Economic Review}, 8(3):351--371.

\bibitem[Aziz and Shah, 2021]{aziz2021participatory}
Aziz, H. and Shah, N. (2021).
\newblock Participatory budgeting: Models and approaches.
\newblock In Rudas, T. and N{\'e}meth, G., editors, {\em Pathways between Social Science and Computational Social Science: Theories, Methods and Interpretations}. Springer.

\bibitem[Baron and Ferejohn, 1989]{baron1989}
Baron, D.~P. and Ferejohn, J.~A. (1989).
\newblock Bargaining in {{Legislatures}}.
\newblock {\em The American Political Science Review}, 83(4):1181--1206.

\bibitem[Baron and Myerson, 1982]{baron1982a}
Baron, D.~P. and Myerson, R.~B. (1982).
\newblock Regulating a {{Monopolist}} with {{Unknown Costs}}.
\newblock {\em Econometrica}, 50(4):911--930.

\bibitem[Battaglini and Coate, 2007]{battaglini2007}
Battaglini, M. and Coate, S. (2007).
\newblock Inefficiency in {{Legislative Policymaking}}: {{A Dynamic Analysis}}.
\newblock {\em American Economic Review}, 97(1):118--149.

\bibitem[Battaglini and Coate, 2008]{battaglini2008}
Battaglini, M. and Coate, S. (2008).
\newblock A {{Dynamic Theory}} of {{Public Spending}}, {{Taxation}}, and {{Debt}}.
\newblock {\em The American Economic Review}, 98(1):201--236.

\bibitem[Biais et~al., 2019]{biais2019blockchain}
Biais, B., Bisi{\`e}re, C., Bouvard, M., and Casamatta, C. (2019).
\newblock The blockchain folk theorem.
\newblock {\em Review of Financial Studies}, 32(5):1662--1715.

\bibitem[Bierbrauer and Hellwig, 2016]{bierbrauer2016}
Bierbrauer, F.~J. and Hellwig, M.~F. (2016).
\newblock Robustly {{Coalition-Proof Incentive Mechanisms}} for {{Public Good Provision}} are {{Voting Mechanisms}} and {{Vice Versa}}.
\newblock {\em The Review of Economic Studies}, 83(4):1440--1464.

\bibitem[B{\"o}rgers et~al., 2015]{börgers2015}
B{\"o}rgers, T., Kr{\"a}hmer, D., and Strausz, R. (2015).
\newblock {\em An Introduction to the Theory of Mechanism Design}.
\newblock {Oxford University Press}.

\bibitem[Bowen et~al., 2014]{bowen2014}
Bowen, T.~R., Chen, Y., and Eraslan, H. (2014).
\newblock Mandatory versus {{Discretionary Spending}}: {{The Status Quo Effect}}.
\newblock {\em American Economic Review}, 104(10):2941--2974.

\bibitem[Budish, 2025]{budish2025trust}
Budish, E. (2025).
\newblock Trust at scale: The economic limits of cryptocurrencies and blockchains.
\newblock {\em Quarterly Journal of Economics}, 140(1):1--62.

\bibitem[Catalini and Gans, 2020]{catalini2020simple}
Catalini, C. and Gans, J.~S. (2020).
\newblock Some simple economics of the blockchain.
\newblock {\em Communications of the ACM}, 63(7):80--90.

\bibitem[Cong and He, 2019]{cong2019blockchain}
Cong, L.~W. and He, Z. (2019).
\newblock Blockchain disruption and smart contracts.
\newblock {\em Review of Financial Studies}, 32(5):1754--1797.

\bibitem[Doval and Skreta, 2022]{doval2022}
Doval, L. and Skreta, V. (2022).
\newblock Mechanism {{Design With Limited Commitment}}.
\newblock {\em Econometrica}, 90(4):1463--1500.

\bibitem[Dragu and Laver, 2017]{dragu2017}
Dragu, T. and Laver, M. (2017).
\newblock Legislative coalitions with incomplete information.
\newblock {\em Proceedings of the National Academy of Sciences}, 114(11):2876--2880.

\bibitem[Dworczak and Muir, 2025]{dworczak2025property}
Dworczak, P. and Muir, E.~V. (2025).
\newblock A mechanism-design approach to property rights.
\newblock Working paper, Northwestern University and MIT Sloan.

\bibitem[Figueroa and Skreta, 2009]{figueroa_role_2009}
Figueroa, N. and Skreta, V. (2009).
\newblock The role of optimal threats in auction design.
\newblock {\em Journal of Economic Theory}, 144(2):884--897.

\bibitem[Fleckinger, 2007]{Fleckinger2007}
Fleckinger, P. (2007).
\newblock Informed principal and countervailing incentives.
\newblock {\em Economics Letters}, 94(2):240--244.

\bibitem[Gersbach et~al., 2023]{pitsuwan2023}
Gersbach, H., Pitsuwan, F., and Valvassori~Bolgè, G. (2023).
\newblock Volatility and {{Resilience}} of {{Democratic}} {{Public-good}} {{Provision}}.

\bibitem[Green and Laffont, 1979]{green1979a}
Green, J. and Laffont, J.-J. (1979).
\newblock {\em Incentives in {{Public Decision Making}}}.
\newblock {North-Holland}, {Amsterdam}.

\bibitem[G{\"u}th and Hellwig, 1986]{guth1986}
G{\"u}th, W. and Hellwig, M. (1986).
\newblock The private supply of a public good.
\newblock In {\em Journal of Economics, Supplement 5}, pages 121--159. Springer.

\bibitem[Hagen and Hernando-Veciana, 2021]{hagen_multidimensional_2021}
Hagen, M. and Hernando-Veciana, A. (2021).
\newblock Multidimensional bargaining and posted prices.
\newblock {\em Journal of Economic Theory}, 196:105317.

\bibitem[Hagerty and Rogerson, 1987]{hagerty_robust_1987}
Hagerty, K.~M. and Rogerson, W.~P. (1987).
\newblock Robust trading mechanisms.
\newblock {\em Journal of Economic Theory}, 42(1):94--107.

\bibitem[Halaburda et~al., 2022]{halaburda2022microeconomics}
Halaburda, H., Haeringer, G., Gans, J.~S., and Gandal, N. (2022).
\newblock The microeconomics of cryptocurrencies.
\newblock {\em Journal of Economic Literature}, 60(3):971--1013.

\bibitem[Hellwig, 2003]{hellwig2003}
Hellwig, M.~F. (2003).
\newblock Public-good provision with many participants.
\newblock {\em Review of Economic Studies}, 70(3):589--614.

\bibitem[Jullien, 2000]{jullien2000}
Jullien, B. (2000).
\newblock Participation {{Constraints}} in {{Adverse Selection Models}}.
\newblock {\em Journal of Economic Theory}, 93(1):1--47.

\bibitem[Krishna and Maenner, 2001]{krishna2001}
Krishna, V. and Maenner, E. (2001).
\newblock {Convex Potentials} with an {Application} to {Mechanism Design}.
\newblock {\em Econometrica}, 69(4):1113--1119.

\bibitem[Kuzmics and Steg, 2017]{kuzmics2017}
Kuzmics, C. and Steg, J.-H. (2017).
\newblock On public good provision mechanisms with dominant strategies and balanced budget.
\newblock {\em Journal of Economic Theory}, 170:56--69.

\bibitem[Laffont, 2000]{laffont2000b}
Laffont, J.-J. (2000).
\newblock {\em Institutions, {{Regulation}} and {{Development}}}.
\newblock {Barbara Weinstock Lectures on the Morals of Trade - Berkeley Graduate Lectures}.

\bibitem[Laffont, 2001]{laffont2001}
Laffont, J.-J. (2001).
\newblock {\em Incentives and Political Economy}.
\newblock {Oxford University Press}.

\bibitem[Lewis and Sappington, 1989]{lewis1989}
Lewis, T.~R. and Sappington, D. E.~M. (1989).
\newblock Countervailing incentives in agency problems.
\newblock {\em Journal of Economic Theory}, 49(2):294--313.

\bibitem[Maggi and {Rodriguez-Clare A.}, 1995]{maggi1995}
Maggi, G. and {Rodriguez-Clare A.} (1995).
\newblock On {{Countervailing Incentives}}.
\newblock {\em Journal of Economic Theory}, 66(1):238--263.

\bibitem[Mailath and Postlewaite, 1990]{mailath1990}
Mailath, G.~J. and Postlewaite, A. (1990).
\newblock Asymmetric information bargaining problems with many agents.
\newblock {\em Review of Economic Studies}, 57(3):351--367.

\bibitem[Myerson, 1979]{myerson1979a}
Myerson, R.~B. (1979).
\newblock Incentive {{Compatibility}} and the {{Bargaining Problem}}.
\newblock {\em Econometrica}, 47(1):61--73.

\bibitem[Myerson, 1983]{myerson1983}
Myerson, R.~B. (1983).
\newblock Mechanism {{Design}} by an {{Informed Principal}}.
\newblock {\em Econometrica}, 51(6):1767--1797.

\bibitem[Mylovanov and Tr{\"o}ger, 2014]{mylovanov2014}
Mylovanov, T. and Tr{\"o}ger, T. (2014).
\newblock Mechanism {{Design}} by an {{Informed Principal}}: {{Private Values}} with {{Transferable Utility}}.
\newblock {\em The Review of Economic Studies}, 81(4 (289)):1668--1707.

\bibitem[Norman, 2004]{norman2004}
Norman, P. (2004).
\newblock Efficient mechanisms for public goods with use exclusions.
\newblock {\em Review of Economic Studies}, 71(4):1163--1188.

\bibitem[Rob, 1989]{rob1989}
Rob, R. (1989).
\newblock Pollution claim settlements under private information.
\newblock {\em Journal of Economic Theory}, 47(2):307--333.

\bibitem[Rochet and Chon{\'e}, 1998]{rochet1998}
Rochet, J.-C. and Chon{\'e}, P. (1998).
\newblock Ironing, {{Sweeping}}, and {{Multidimensional Screening}}.
\newblock {\em Econometrica}, 66(4):783--826.

\bibitem[Scheuer and Wolitzky, 2016]{scheuer}
Scheuer, F. and Wolitzky, A. (2016).
\newblock Capital taxation under political constraints.
\newblock {\em American Economic Review}, 106(8):2304--28.

\bibitem[Teichgr{\"a}ber et~al., 2022]{teichgraber2022}
Teichgr{\"a}ber, J., {\v Z}u{\v z}ek, S., and Hensel, J. (2022).
\newblock Optimal {{Short-Time Work}}: {{Screening}} for {{Jobs}} at {{Risk}}.

\end{thebibliography}

\end{document}